\documentclass[preprint,12pt]{elsarticle}

\usepackage{amssymb}
\usepackage{amsmath}
\usepackage{amsthm}
\usepackage{mathtools}

\usepackage{graphicx}
\usepackage{float}
\usepackage{multirow}
\usepackage{placeins}
\usepackage{dblfloatfix}

\usepackage{algorithm}
\usepackage{algpseudocode}

\usepackage{xcolor}
\usepackage[unicode,colorlinks,linkcolor=black,citecolor=black,urlcolor=blue,hypertexnames=false]{hyperref}

\theoremstyle{definition}
\newtheorem{definition}{Definition}
\newtheorem{assumption}{Assumption}

\theoremstyle{plain}
\newtheorem{theorem}{Theorem}
\newtheorem{proposition}{Proposition}

\journal{Aerospace Science and Technology}

\begin{document}

\begin{frontmatter}

\title{Decentralized Affine Transformation for Scalable and Safe Multi-Agent Aerial Coordination}

\author[1]{Garegin Mazmanyan}
\ead{gmazmanyan@arizona.edu}

\author[2,3]{Hossein Rastgoftar\corref{cor1}}
\cortext[cor1]{Corresponding author}
\ead{hrastgoftar@arizona.edu}

\affiliation[1]{organization={Department of Computer Science, University of Arizona},
            city={Tucson},
            postcode={85721},
            state={AZ},
            country={USA}}

\affiliation[2]{organization={Department of Aerospace and Mechanical Engineering, University of Arizona},
            city={Tucson},
            postcode={85721},
            state={AZ},
            country={USA}}

\affiliation[3]{organization={Department of Electrical and Computer Engineering, University of Arizona},
            city={Tucson},
            postcode={85721},
            state={AZ},
            country={USA}}

\begin{abstract}
This paper presents an experimental evaluation of decentralized affine transformation (AT) in multi-agent systems using teams of mini-quadcopters. The AT framework enables an agent team to safely navigate constrained environments with narrow passages while allowing aggressive changes in inter-agent distances, which are formally characterized through the decomposition of the AT transformation matrix. We focus on two-dimensional AT, formulated as a decentralized leader--follower problem. In this formulation, three leader quadcopters are positioned at the vertices of a triangle, while all follower quadcopters remain within the triangle. The leaders know the desired trajectories prescribed by the AT, whereas the followers do not. Instead, the followers infer their trajectories through local communication governed by fixed communication weights determined by the team's initial spatial configuration. Experimental results provide evidence consistent with the predicted bounded-error tracking and convergence of decentralized AT and demonstrate its capability to safely guide multi-agent teams through constrained environments with narrow passages.
\end{abstract}

\begin{highlights}
\item Experimental validation of decentralized affine transformation for UAV swarms
\item Single principal-strain bound certifies collision-free motion for any team size
\item Followers use only neighbors' measured positions, no leader trajectories
\item Bounded-error tracking on six quadrotors flying through a narrow passage
\end{highlights}


\end{frontmatter}

%% ============================================================
%% Main text
%% ============================================================

\section{Introduction}

Safe and scalable coordination of multi-agent aerial systems in constrained environments remains a fundamental challenge, particularly when formations must undergo significant geometric deformation while maintaining guaranteed inter-agent separation. Existing decentralized coordination frameworks, including consensus and containment control, have been extensively studied and applied across multi-agent systems \cite{chang2025consensus, ma2024consensus, thummalapeta2023survey, li2021containment}. While these approaches provide convergence guarantees under various communication and dynamical assumptions, they lack mechanisms to enforce safety under aggressive reconfiguration. In contrast, safety-critical control methods that impose pairwise collision-avoidance constraints scale quadratically with the number of agents, limiting their applicability to large teams.

This paper introduces a control-theoretic framework for scalable and safe multi-agent coordination based on affine transformation (AT) and principles of continuum deformation. The key idea is to represent the collective motion of an agent team as a global affine transformation, whose deformation is characterized by principal strains. By constraining these strains, safety can be enforced through a single global condition that is independent of the number of agents. This provides a fundamentally different safety abstraction compared to pairwise constraint-based methods.

The framework further admits a decentralized realization through a leader--follower architecture in which follower agents reconstruct their desired trajectories using fixed barycentric weights and local measurements of neighboring agents. Unlike existing approaches, this realization does not require access to global trajectories or model-based prediction, and operates purely on real-time measured states. The resulting closed-loop system couples local tracking dynamics with the communication topology, yielding convergence of all agents to a globally consistent affine transformation.

The central contribution of this paper is to show that affine transformation serves as both a geometric description of collective motion and a scalable, safety-certified coordination mechanism for multi-agent aerial systems. Unlike pairwise safety methods, the proposed framework encodes collision avoidance through global deformation constraints on the principal strains of the transformation. Unlike classical containment or consensus-based coordination, followers do not receive leader trajectories or global reference signals; instead, they reconstruct the affine motion through local measurements and fixed barycentric communication weights. This combination of global deformation planning, decentralized acquisition, and hardware validation provides a practical framework for safe multi-agent aerial coordination in constrained environments.

\subsection{Related work}
Consensus and containment control represent two foundational paradigms in decentralized multi-agent coordination. Consensus protocols, in particular, have been extensively studied and widely applied across diverse domains, including formation control of unmanned vehicles \cite{chang2025consensus, ma2024consensus}, distributed sensing and coverage \cite{facinelli2019gas, pasek2022review}, flocking and swarming dynamics \cite{fernando2021flocking, epp2024ota}, cooperative decision-making \cite{stankovic2023tdlearning, bassolillo2025mh}, and network load balancing \cite{brooks1996sensorfusion, pasek2022review}. Their stability and convergence properties have been rigorously characterized for communication networks without delays and those subject to delays \cite{olfati2004consensus, li2010consensus, li2010distributed}.

Containment control has been extensively studied over the last decade
\cite{thummalapeta2023survey,haghshenas2015containment,li2021containment,bechlioulis2024robust,xu2025containment}.
Early work on containment control investigated distributed containment under stationary and dynamic leaders \cite{wang2010distributed}, introduced the concept of robust containment control \cite{su2011robust}, developed decentralized leader--follower coordination under fixed and switching communication topologies \cite{ni2010leader}, and proposed a model-reference adaptive control framework \cite{meng2010distributed}. In \cite{thummalapeta2023survey}, the authors comprehensively review key concepts, communication frameworks, dynamics modeling, and controller design strategies in containment control. Containment control of heterogeneous linear multi-agent systems has been studied in \cite{wang2010distributed,haghshenas2015containment}. Containment control guided by multiple leaders in the presence of input constraints is addressed in \cite{li2021containment}. Assuming each agent is modeled by high-order dynamics, containment control is demonstrated in \cite{meng2010distributed} under nonconvex inputs and delays, while accounting for realistic constraints. Safety and collision avoidance is one main challenge of multi-agent coordination to ensure scalability. Recent learning-based methods obtain safe UAV control through learned policies with formal guarantees: safe reinforcement learning from human demonstration has been applied to adaptive hierarchical quadcopter control \cite{tan2026adaptive}, and fixed-time stochastic learning has been developed for human--UAV interaction under state and input constraints \cite{tan2026fixed}. These methods enforce safety at the level of a single vehicle or a human--vehicle pair, and are complementary to the team-level geometric safety certificate pursued in this work.

The corresponding author has extensively investigated the concept of decentralized AT over the past decade \cite{rastgoftar2021safe, rastgoftar2022integration, rastgoftar2021scalable}. Inspired by the principles of continuum mechanics, AT has also been referred to as a \textit{homogeneous transformation} in previous work, characterized by a Jacobian matrix and a rigid-body translation. By applying the polar decomposition of the Jacobian matrix into rotation, shear deformation, and axial deformation, formal guarantees for agent containment and inter-agent collision avoidance are provided by constraining the axial deformation of the homogeneous transformation \cite{rastgoftar2021safe, rastgoftar2021scalable}.

Similar to containment control, AT can also be implemented within a decentralized leader--follower framework. Unlike traditional methods, however, AT explicitly characterizes safety and guarantees collision avoidance, while remaining inherently scalable to large numbers of agents through purely local communication. In particular, AT in $\mathbb{R}^n$ can be formulated as a decentralized leader--follower problem, where the trajectories of $n+1$ leaders positioned at the vertices of an $n$-D simplex are tracked and adopted by the followers. Thus, AT reduces to a leader path-planning problem that has been effectively addressed using classical optimal control \cite{rastgoftar2018cooperative}, A* search \cite{rastgoftar2022integration, rastgoftar2021scalable}, and particle swarm optimization \cite{liang2019multi}.

Pairwise collision-avoidance constraints grow quadratically with the number of agents and can become infeasible in large teams. In contrast, decentralized AT enforces safety for an arbitrary number of agents in a computationally efficient manner, as a large population of followers can acquire the desired AT solely through local communication at minimal computational cost.

Centralized AT, where each agent is provided with the desired trajectories defined by an AT, has been experimentally validated in \cite{romano2019experimental, 10156556}. Decentralized AT was also explored in \cite{romano2019experimental}, under the restriction to uniform contraction; however, a comprehensive experimental validation has remained an open challenge.

The present work differs from prior studies on affine or homogeneous transformation in three important ways. First, while earlier work primarily developed theoretical formulations or centralized implementations, this paper focuses on decentralized acquisition in which followers use only real-time measured neighbor positions and fixed communication weights. Second, although previous experimental studies considered restricted deformation patterns, the present work validates decentralized affine transformation under contraction, rigid-body motion, and non-uniform deformation in a physically constrained environment. Third, this paper explicitly connects the deformation parameters to a global collision-avoidance condition and evaluates the resulting safety margin using hardware flight data. Therefore, the contribution is not merely an experimental replication of prior AT methods, but a decentralized and safety-certified realization of affine coordination for aerial robotic teams.

\subsection{Contributions}

This paper develops and experimentally evaluates a decentralized AT framework for safe multi-agent aerial coordination. The main contributions are as follows.

\begin{enumerate}
    \item \textbf{Global safety characterization through affine deformation:}
    We establish a collision-avoidance condition based on the principal strains of the affine transformation. The condition depends only on the minimum initial inter-agent separation, the agent radius, and the tracking-error bound, and provides a global safety certificate that does not require pairwise constraints.

    \item \textbf{Decentralized acquisition without trajectory sharing:}
    We formulate a leader--follower realization in which each follower reconstructs its reference trajectory using only real-time measured positions of its neighbors and fixed barycentric weights determined from the initial configuration. Under local tracking convergence and fixed leader terminal positions, this closed-loop architecture drives all agents to positions consistent with the prescribed affine transformation, without access to leader trajectories, global desired positions, or online optimization.

    \item \textbf{Hardware validation in constrained environments:}
    We validate the framework using six quadcopters navigating through a narrow passage, demonstrating formation contraction, rigid-body motion, and non-uniform deformation with bounded tracking errors and measured inter-agent separation above the collision threshold.
\end{enumerate}\par
\subsection{Outline}\label{sec:outline}
This paper is organized as follows: The preliminary notions and properties of AT are reviewed in Section~\ref{Prem}. Assuming that each agent is modeled by generic dynamics, AT planning and acquisition are presented in Section~\ref{Approach}. The experimental setup is described in Section~\ref{sec:experimental_setup}, and the experimental results of decentralized AT are reported in Section~\ref{sec:results}. A discussion of the results and limitations is provided in Section~\ref{sec:discussion}. Finally, the concluding remarks are provided in Section~\ref{Conclusion}.

\section{Preliminaries}\label{Prem}
We consider a group of $N$ quadcopters defined by set $\mathcal{V}=\left\{u_1,\cdots,u_N\right\}$. Set $\mathcal{V}$ can be expressed as $\mathcal{V}=\mathcal{L}\cup \mathcal{F}$, where $\mathcal{L}=\left\{u_1,u_2,u_3\right\}$ and $\mathcal{F}=\left\{u_4,\cdots,u_N\right\}$ define  the leader and follower quadcopters, respectively. Leaders are at vertices of a triangle that is called the \textit{leading triangle} which is deformable and its deformation is determined by an AT at any time $t$.  The remaining followers are desired to be contained inside the leading triangle. Inter-agent communications among the quadcopters are structured by a directed graph $\mathcal{G}\left(\mathcal{V},\mathcal{E}\right)$, where $\mathcal{E}\subset \mathcal{V}\times \mathcal{V}$ defines inter-agent communication links. We say that $(j,i)\in \mathcal{E}$, if $i\in \mathcal{V}$ communicates with $j\in \mathcal{V}$. Given set $\mathcal{E}$,
\begin{equation}
    \mathcal{N}_i=\left\{j\in\mathcal{V}:\left(j,i\right)\in \mathcal{E} \right\},\qquad \forall i\in \mathcal{V},
\end{equation}
defines in-neighbor quadcopters of quadcopter $i\in \mathcal{V}$.
\begin{assumption}\label{asum1}
    Leader quadcopters move independently, therefore,
    \begin{equation}
        \bigwedge_{i\in \mathcal{L}}\left(\mathcal{N}_i=\emptyset\right).
    \end{equation}
\end{assumption}
\begin{assumption}\label{structurconst}
    This paper considers a two-dimensional AT of a multi-quadcopter system where every follower quadcopter is strictly communicates with three in-neighbors. This assumption can be formally specified as follows:
    \begin{equation}
        \bigwedge_{i\in \mathcal{F}}\left(\left|\mathcal{N}_i\right|=3\right).
    \end{equation}
\end{assumption}
Assumption \ref{structurconst} will be used for obtaining unique communication weights for every follower $i\in \mathcal{V}$ and  prove the decentralized convergence of AT coordination.

For every agent $i\in \mathcal{V}$,
\begin{subequations}
    \begin{equation}
        \mathbf{a}_i=\begin{bmatrix}X_i&Y_i&Z\end{bmatrix}^T,\qquad i\in \mathcal{V},
    \end{equation}
    \begin{equation}        \mathbf{r}_i(t)=\begin{bmatrix}x_i(t)&y_i(t)&z_i(t)\end{bmatrix}^T,\qquad i\in \mathcal{V},
    \end{equation}
    \begin{equation}        \mathbf{r}_{i,d}(t)=\begin{bmatrix}x_{i,d}(t)&y_{i,d}(t)&z_{i,d}(t)\end{bmatrix}^T,\qquad i\in \mathcal{V},
    \end{equation}
    \begin{equation}        \mathbf{p}_i(t)=\begin{bmatrix}p_{x,i}(t)&p_{y,i}(t)&p_{z,i}(t)\end{bmatrix}^T,\qquad i\in \mathcal{V},
    \end{equation}
\end{subequations}
denote the initial position, the actual position, the reference position (input to the control system), and the global desired position, respectively. In this paper, the actual position of every quadcopter $i\in \mathcal{V}$ is precisely measured by a high-resolution motion capture system, with the specifics provided in Section~\ref{sec:experimental_setup}. Note that the $Z$ component of the reference positions of all the agents are the same as $Z$ which in turn implies that the reference configuration of the quadcopter team is distributed in a horizontal plane.

The global desired position of every agent $i\in \mathcal{V}$ is defined by
\begin{equation}\label{affinetransformation}
\mathbf{p}_i(t)=\mathbf{Q}(t)\mathbf{a}_i+\mathbf{d}(t),\qquad \forall i\in \mathcal{V},
\end{equation}
where $\mathbf{Q}\in \mathbb{R}^{3\times 3}$ is the Jacobian matrix and  $\mathbf{d}\in \mathbb{R}^{3\times 1}$ is the  rigid-body translation vector.

\begin{theorem}
    Under transformation, the global desired position $\mathbf{p}_i(t)$ of every follower $i\in \mathcal{F}$ is uniquely specified based on leaders' desired position at any time $t$ by
    \begin{equation}\label{leader-followerdesired}
        \mathbf{p}_i(t)=\sum_{j\in \mathcal{L}}\alpha_{i,j}\mathbf{p}_j(t),\qquad \forall i\in \mathcal{F},
    \end{equation}
    where $\alpha_{i,u_1}$, $\alpha_{i,u_2}$, and $\alpha_{i,u_3}$ are constant and obtained based on reference position of leaders and $i\in \mathcal{F}$ by
    \begin{equation}
        \begin{bmatrix}
            \alpha_{i,u_1}\\
            \alpha_{i,u_2}\\
            \alpha_{i,u_3}
        \end{bmatrix}
        =\begin{bmatrix}
        X_{u_1}&X_{u_2}&X_{u_3}\\
        Y_{u_1}&Y_{u_2}&Y_{u_3}\\
        1&1&1
        \end{bmatrix}
        ^{-1}
        \begin{bmatrix}
        X_i\\
        Y_i\\
        1
        \end{bmatrix}
        .
    \end{equation}
\end{theorem}
\begin{proof}

The three leaders are non-collinear, so the matrix in the statement is invertible and the weights $\alpha_{i,u_1},\alpha_{i,u_2},\alpha_{i,u_3}$ are the unique solution of the linear system; its last row gives $\sum_{j\in\mathcal{L}}\alpha_{i,j}=1$. Because every agent shares the same height $Z$, the same weights satisfy $\mathbf{a}_i=\sum_{j\in\mathcal{L}}\alpha_{i,j}\mathbf{a}_j$. Applying \eqref{affinetransformation} and using $\sum_{j\in\mathcal{L}}\alpha_{i,j}=1$,
\[
\mathbf{p}_i=\mathbf{Q}\mathbf{a}_i+\mathbf{d}
=\sum_{j\in\mathcal{L}}\alpha_{i,j}\left(\mathbf{Q}\mathbf{a}_j+\mathbf{d}\right)
=\sum_{j\in\mathcal{L}}\alpha_{i,j}\mathbf{p}_j,
\]
which is \eqref{leader-followerdesired}. The weights are constant because they depend only on the initial positions. See also \cite{rastgoftar2022integration, rastgoftar2021safe, rastgoftar2021scalable}.

\end{proof}

The reference trajectory of agent $i\in \mathcal{V}$ is defined by
\begin{equation}\label{desiredtraject}
\mathbf{r}_{i,d}(t)=\begin{cases}
\mathbf{p}_i(t)&i\in \mathcal{L}\\
\sum_{j\in \mathcal{N}_i}w_{i,j}\mathbf{r}_j(t)&i\in \mathcal{F}
\end{cases}
,
\end{equation}
where $w_{i,j}$ is positive and constant, for every $i\in \mathcal{F}$ and $j\in \mathcal{N}_i$, and
\begin{equation}
    \sum_{j\in \mathcal{N}_i}w_{i,j}=1.
\end{equation}
To achieve an AT in a decentralized fashion, $w_{i,j}$ is constant and determined based on the reference position of $i\in \mathcal{F}$ and its in-neighbors, defined by $\mathcal{N}_i$. Expressing $\mathcal{N}_i=\left\{i_1,i_2,i_3\right\}$, communication weights of every follower $i\in \mathcal{F}$ is determined by
    \begin{equation}\label{folcomweights}
        \begin{bmatrix}
            w_{i,i_1}\\
            w_{i,i_2}\\
            w_{i,i_3}
        \end{bmatrix}
        =\begin{bmatrix}
        X_{i_1}&X_{i_2}&X_{i_3}\\
        Y_{i_1}&Y_{i_2}&Y_{i_3}\\
        1&1&1
        \end{bmatrix}
        ^{-1}
        \begin{bmatrix}
        X_i\\
        Y_i\\
        1
        \end{bmatrix}
    \end{equation}
\begin{definition}
    We define weight matrix $\mathbf{W}=\left[W_{ij}\right]\in \mathbb{R}^{N\times N}$ as the weighted Laplacian matrix with the $(i,j)$ entry of $\mathbf{W}$ is defined as follows:
    \begin{equation}
    W_{ij}=\begin{cases}
        -1&j=i\\
        w_{u_i,u_j}&u_i\in \mathcal{F},~u_j\in \mathcal{N}_{u_i}\\
        0&\mathrm{otherwise}
    \end{cases}
    .
    \end{equation}
\end{definition}
\begin{definition}
    We define
    \begin{equation}
        \mathbf{L}=\begin{bmatrix}
            \mathbf{I}_3&\mathbf{0}_{3\times (N-3)}
        \end{bmatrix}^T\in \mathbb{R}^{N\times 3}.
    \end{equation}
\end{definition}
\begin{definition}
    We define non-negative matrix $\mathbf{H}=\left[H_{ij}\right]\in \mathbb{R}^{N\times 3}$ with the $(i,j)$ entry
    \begin{equation}
        H_{ij}=\begin{cases}
            1&u_i\in \mathcal{L},~j=i\\
            \alpha_{u_i,u_j}&u_i\in\mathcal{F},~u_j\in \mathcal{L}\\
            0&\mathrm{otherwise}
        \end{cases}
        .
    \end{equation}
\end{definition}
\begin{theorem}\label{thm2}
    Assume inter-agent communication is defined such that there exists at least one path from every leader $j\in \mathcal{L}$ to every follower $i\in \mathcal{F}$ and Assumptions \ref{asum1} and \ref{structurconst} are satisfied.
    Then, $\mathbf{W}\in \mathbb{R}^{N\times N}$ is Hurwitz and
    \begin{equation}
        \mathbf{H}=-\mathbf{W}^{-1}\mathbf{L}.
    \end{equation}
\end{theorem}
\begin{proof}

Order the agents so that the three leaders come first. By Assumption~\ref{asum1} each leader row of $\mathbf{W}$ carries $-1$ on the diagonal and zeros elsewhere, hence
\[
\mathbf{W}=\begin{bmatrix}-\mathbf{I}_3 & \mathbf{0}\\[2pt] \mathbf{W}_{FL} & \mathbf{W}_{FF}\end{bmatrix},
\]
and the eigenvalues of $\mathbf{W}$ are $-1$ with multiplicity three together with those of $\mathbf{W}_{FF}$. Write $\mathbf{W}_{FF}=\mathbf{M}_{FF}-\mathbf{I}$, where $\mathbf{M}_{FF}\ge \mathbf{0}$ holds the follower-to-follower weights. By \eqref{folcomweights} each follower's weights are positive and sum to one, so row $i$ of $\mathbf{M}_{FF}$ sums to $1-s_i$, where $s_i\ge 0$ is the weight follower $i$ assigns to its leader in-neighbors; $\mathbf{M}_{FF}$ is therefore substochastic. The hypothesis that a path runs from a leader to every follower makes $\mathbf{M}_{FF}$ convergent, so $\rho(\mathbf{M}_{FF})<1$ and every eigenvalue of $\mathbf{W}_{FF}$ has real part at most $\rho(\mathbf{M}_{FF})-1<0$. Thus $\mathbf{W}$ is Hurwitz and invertible. Evaluating \eqref{desiredtraject} at the affine configuration gives $\mathbf{W}\mathbf{X}+\mathbf{L}\mathbf{X}_{L,d}=\mathbf{0}$, so $\mathbf{X}=-\mathbf{W}^{-1}\mathbf{L}\mathbf{X}_{L,d}$; matching $\mathbf{X}=\mathbf{H}\mathbf{X}_{L,d}$ for every leader configuration yields $\mathbf{H}=-\mathbf{W}^{-1}\mathbf{L}$. See also \cite{rastgoftar2022integration, rastgoftar2021safe, rastgoftar2021scalable}.

\end{proof}

\section{Approach}\label{Approach}
We decompose the AT of a quadcopter team into two components: (i) continuum deformation planning, which specifies a safe global transformation, and (ii) decentralized acquisition, through which agents realize the transformation using local feedback and communication.

\subsection{AT planning}\label{Affine Transformation Planning}

As discussed in Section~\ref{Prem}, AT can be formulated as a decentralized leader--follower problem, where the desired trajectories of the leaders are prescribed by~\eqref{affinetransformation}, and the followers acquire their trajectories through local inter-agent communication. Leader agents have knowledge of the transformation parameters $(\mathbf{Q},\mathbf{d})$, whereas follower agents do not have access to the global desired trajectories $\mathbf{p}_i(t)$ and instead converge to them using only local information. Under an AT, inter-agent distances can vary significantly. Therefore, the transformation parameters must be designed to ensure safety during deformation. While feasible translational motion $\mathbf{d}(t)$ can be obtained using standard motion planning methods \cite{rastgoftar2022integration, rastgoftar2021safe}, the key design variable for safety is the deformation matrix $\mathbf{Q}(t)$.

Using polar decomposition, $\mathbf{Q}$ is expressed as
\begin{equation}
\mathbf{Q}=\mathbf{R}_r\mathbf{U},
\end{equation}
where $\mathbf{R}_r\in \mathbb{R}^{3\times3}$ is a rotation matrix and $\mathbf{U}\in \mathbb{R}^{3\times 3}$ is symmetric positive definite. The matrix $\mathbf{U}$ admits the decomposition
\begin{equation}
\mathbf{U}=\mathbf{R}_D\mathbf{\Lambda}\mathbf{R}_D^T,
\end{equation}
where $\mathbf{R}_D$ specifies shear deformation and
\begin{equation}
\mathbf{\Lambda}=\mathrm{diag}(\lambda_1,\lambda_2,1)
\end{equation}
defines the principal strains. The parameters $\lambda_1,\lambda_2$ govern contraction and expansion along principal directions and directly determine inter-agent spacing under deformation.

For the planar experimental setting, we restrict to $\mathbb{R}^2$ with $\phi_r=\theta_r=0$, and parameterize the transformation using
\[
\{d_1(t),d_2(t),\lambda_1(t),\lambda_2(t),\psi_d(t),\psi_r(t)\}.
\]

\medskip
\noindent\textbf{Safety characterization via principal strains.}
To ensure collision-free coordination, the deformation must preserve sufficient separation between agents despite tracking errors. The following theorem provides a sufficient condition for safety.

\begin{theorem}[Collision avoidance under affine deformation]
\label{thm:safety}
Consider the affine transformation defined by \eqref{affinetransformation}. Suppose:

\begin{enumerate}
    \item $\|\mathbf{a}_i-\mathbf{a}_j\|\ge d_{\min}$ for all $i\neq j$,
    \item each agent has radius $r>0$,
    \item $\|\mathbf{r}_i(t)-\mathbf{p}_i(t)\|\le \delta$ for all $i$ and $t$,
    \item $\lambda_1(t),\lambda_2(t)\ge \lambda_{\min}$ for all $t$,
\end{enumerate}

where
\[
\lambda_{\min}=\frac{2(\delta+r)}{d_{\min}}.
\]

Then, for all $t\ge 0$,
\[
\|\mathbf{r}_i(t)-\mathbf{r}_j(t)\|\ge 2r,\qquad \forall i\neq j,
\]
and no collisions occur.
\end{theorem}

\begin{proof}
Under affine transformation,
\[
\|\mathbf{p}_i-\mathbf{p}_j\| = \|\mathbf{Q}(\mathbf{a}_i-\mathbf{a}_j)\|
\ge \lambda_{\min}\|\mathbf{a}_i-\mathbf{a}_j\|
\ge \lambda_{\min} d_{\min}.
\]
Using the triangle inequality,
\[
\|\mathbf{r}_i-\mathbf{r}_j\|
\ge \|\mathbf{p}_i-\mathbf{p}_j\| - 2\delta
\ge \lambda_{\min} d_{\min} - 2\delta.
\]
Substituting $\lambda_{\min}$ gives $\|\mathbf{r}_i-\mathbf{r}_j\|\ge 2r$.
\end{proof}

Theorem~\ref{thm:safety} shows that collision avoidance can be guaranteed using a single global constraint on the deformation, independent of the number of agents. This provides a scalable alternative to pairwise constraint-based approaches.

\noindent\textit{Remark} (tightness of the bound). The condition is tight only when two nearest agents sit at the minimum deformed separation $\lambda_{\min}d_{\min}$ and, at that instant, deviate from their global desired positions in opposite directions, each by the full $\delta$. A pairwise distance is reduced only by the relative deviation $(\mathbf{r}_i-\mathbf{p}_i)-(\mathbf{r}_j-\mathbf{p}_j)$, so a translation shared by the whole team leaves every separation unchanged. Bounding this relative term by $\Delta\le 2\delta$ sharpens the condition to $\lambda_{\min}=(2r+\Delta)/d_{\min}$, which is the relevant quantity when neighbor tracking errors are correlated. The bound is therefore conservative for teams whose members deform together while sharing a common transient motion, the regime of the experiments in Section~\ref{sec:results}.

\subsection{Affine transformation acquisition}\label{Affine Transformation Acquisition}

Each agent is modeled as
\begin{equation}\label{nonldyn}
\begin{cases}
\dot{\mathbf{x}}_i=\mathbf{f}(\mathbf{x}_i,\mathbf{u}_i)+\mathbf{B}_i\mathbf{G}_i(\mathbf{r}_{i,d}-\mathbf{r}_i),\\
\mathbf{r}_i=\mathbf{C}\mathbf{x}_i,
\end{cases}
\end{equation}
where $\mathbf{x}_i$ is the state and $\mathbf{r}_i$ is the position.

\begin{definition}
The invariant set is
\[
\mathcal{S}_i=\{\mathbf{x}_i:\mathbf{r}_i=\mathbf{r}_{i,d},~\bar{\mathbf{r}}_i=0\}.
\]
\end{definition}

\begin{assumption}\label{finalconst}
Leader trajectories become constant for $t\ge t_f$.
\end{assumption}

\begin{definition} We define
\[
\mathbf{x}_i^*=\begin{bmatrix}\mathbf{p}_i(t_f)\\\bar{\mathbf{p}}_i(t_f)\end{bmatrix}.
\]
as the desired final state of agent $i$.
\end{definition}

% \medskip
% \noindent\textbf{Closed-loop realization of affine transformation.}

\begin{theorem}
If each agent asymptotically tracks its reference $\mathbf{r}_{i,d}$ and Assumption~\ref{finalconst} holds, then $\mathbf{r}_i(t)\to \mathbf{p}_i(t_f)$ for all agents.
\end{theorem}

\begin{proof}
Using the stacked dynamics and communication structure, the tracking errors satisfy
\[
\mathbf{X}_d-\mathbf{X}=\mathbf{W}\mathbf{X}+\mathbf{L}\mathbf{X}_{L,d},
\]
and similarly for $\mathbf{Y}$. As $\mathbf{r}_i\to\mathbf{r}_{i,d}$, we obtain $\mathbf{W}\mathbf{X}+\mathbf{L}\mathbf{X}_{L,d}\to 0$, which implies $\mathbf{X}\to\mathbf{H}\mathbf{X}_{L,d}$. This corresponds to $\mathbf{r}_i\to \mathbf{p}_i$ for all agents.
\end{proof}

\noindent\textit{Remark.}
The closed-loop system couples local tracking dynamics with the communication topology, ensuring that decentralized feedback and local interactions recover the global affine transformation.

The asymptotic result assumes the leaders have stopped. While the formation deforms, the leader trajectories are time-varying and the follower references \eqref{desiredtraject} inherit this motion, so the guarantee during that phase is a bounded error rather than convergence to zero. The following proposition makes this precise.

\begin{proposition}[Bounded tracking during deformation]\label{prop:bounded}
Let $\mathbf{e}(t)=\mathbf{X}(t)-\mathbf{H}\mathbf{X}_{L,d}(t)$ be the deviation of the team from the affine configuration set by the current leader positions. Under \eqref{desiredtraject} this deviation follows
\[
\dot{\mathbf{e}}=\mathbf{W}\mathbf{e}+\mathbf{w}(t),
\]
where $\mathbf{w}(t)$ gathers the leader velocity that followers do not feed forward and the position measurement noise, with $\|\mathbf{w}(t)\|\le\bar{w}$. Because $\mathbf{W}$ is Hurwitz (Theorem~\ref{thm2}), the system is input-to-state stable and
\[
\|\mathbf{e}(t)\|\le\kappa\, e^{-\gamma t}\|\mathbf{e}(0)\|+\frac{\kappa}{\gamma}\,\bar{w},
\qquad \kappa\ge 1,\ \gamma>0,
\]
with $\gamma$ set by the spectral abscissa of $\mathbf{W}$.
\end{proposition}

\noindent\textit{Remark.} While the leaders move, $\bar{w}>0$ and the follower error is ultimately bounded by $(\kappa/\gamma)\bar{w}$. Once the leaders reach their final configuration at $t\ge t_f$ (Assumption~\ref{finalconst}), $\mathbf{w}(t)\to\mathbf{0}$ and $\mathbf{e}(t)\to\mathbf{0}$, which recovers the asymptotic result above. Section~\ref{sec:results} reports follower errors that stay bounded and peak in the fastest leader phase, consistent with this ultimate bound.

The block diagram of the control system for AT of a multi-agent team is shown in Fig.~\ref{ControlBlockDiagram}, where each agent is modeled by generic nonlinear dynamics.
\begin{figure}[t]
    \centering
    \includegraphics[width=\textwidth]{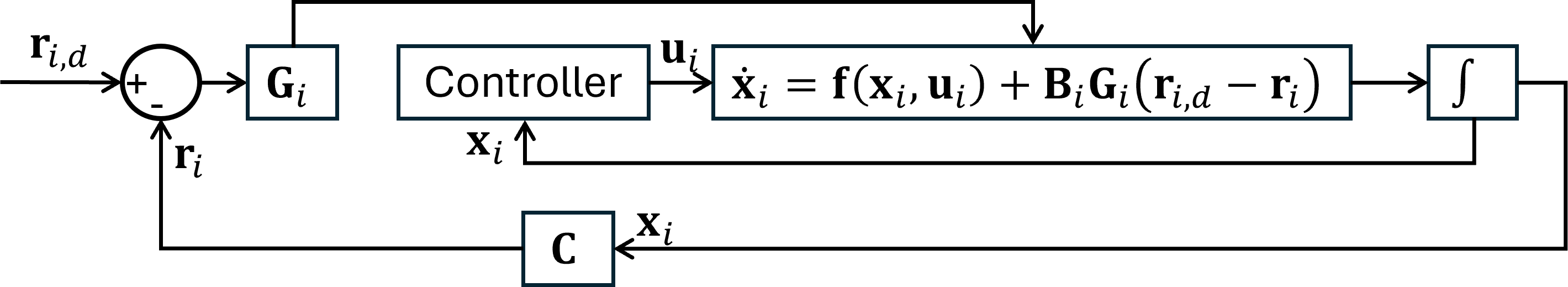}
    \caption{Block diagram of the control system developed for affine formation control of a multi-agent team with generic nonlinear dynamics.}
    \label{ControlBlockDiagram}
\end{figure}
\subsection{Discussion: Distinction from classical containment control}

The proposed decentralized realization of affine transformation differs fundamentally from classical containment control and consensus-based coordination frameworks in both structure and information flow.
First, follower agents do not require access to the desired trajectories of the leaders. Instead, each follower computes its reference signal using only the \emph{actual measured positions} of its neighbors, as defined in \eqref{desiredtraject}. This results in an output-feedback architecture in which the global affine transformation is reconstructed through local interactions, without trajectory sharing, global reference signals, or model-based prediction. Second, the communication weights are fixed and determined a priori from the initial configuration using barycentric coordinates. Unlike adaptive or consensus-based methods, no online weight updates, optimization, or agreement dynamics are required. This preserves the affine structure of the formation throughout the evolution.
Third, the proposed framework does not rely on consensus or agreement dynamics. Rather than driving agents to a common value or to the convex hull of leader trajectories, the method enforces a structured mapping from leader motion to follower motion. As a result, the entire team evolves according to a globally consistent affine transformation, including rotation, shear, and non-uniform scaling.
Finally, the closed-loop system couples local tracking dynamics with the communication topology to yield a decentralized realization of a global deformation. This distinguishes the proposed approach from classical containment control, where convergence is defined relative to convex sets but without an explicit global transformation model. Collectively, these properties enable a scalable decentralized implementation in which each agent operates using only local measurements while contributing to a coordinated motion governed by continuum deformation.

\section{Experimental setup}\label{sec:experimental_setup}

The decentralized affine transformation (AT) framework was experimentally validated using six Crazyflie 2.1 quadrotors in an indoor laboratory at the University of Arizona. A Vicon motion capture system provided high-precision 6-DOF pose estimates at 100~Hz, which were streamed to a central ground station running the Crazyswarm2 framework.

\subsection{Agent configuration}

Among the six agents, cf1, cf2, and cf6 were designated as \textit{leaders} ($\mathcal{L} = \{1, 2, 6\}$), while cf3, cf4, and cf5 acted as \textit{followers} ($\mathcal{F} = \{3, 4, 5\}$). The initial positions were specified in a YAML configuration file, with the formation operating at a nominal altitude of $Z = 0.75$~m. Table~\ref{tab:initial_positions_updated} summarizes the initial configuration.

The role and pose assignment follows a placement rule. The three leaders must occupy non-collinear initial positions, so that the leading triangle is nondegenerate and the matrix inverses in Eqs.~\eqref{leader-followerdesired} and \eqref{folcomweights} exist. Each follower must lie in the interior of the triangle formed by its three in-neighbors, which makes its barycentric weights positive and, by Theorem~\ref{thm2}, renders the weight matrix Hurwitz. The configuration in Table~\ref{tab:initial_positions_updated} satisfies both conditions.

\begin{table}[t]
    \centering
    \begin{tabular}{|c|c|c|c|c|c|c|}
    \hline
         & cf1 & cf2 & cf3 & cf4 & cf5 & cf6 \\
         \hline
         $X_i$ (m) & 0.00 & 1.00 & 0.25 & 0.75 & 0.50 & 0.50 \\
         $Y_i$ (m) & 0.00 & 0.00 & 0.50 & 0.50 & 1.00 & 1.50 \\
         Role & L & L & F & F & F & L \\
         \hline
    \end{tabular}
    \caption{Initial positions and roles of each quadcopter. L = Leader, F = Follower.}
    \label{tab:initial_positions_updated}
\end{table}

\subsection{Decentralized control architecture}

The control architecture implements the decentralized AT framework described in Section~\ref{Prem}. Critically, the system operates without requiring followers to have prior knowledge of the leader trajectories. Instead, each follower computes its desired position at every control timestep using only the \textit{actual measured positions} of its neighbors, as specified by Eq.~\eqref{desiredtraject}:
\begin{equation}
\mathbf{r}_{i,d}(t) = \sum_{j \in \mathcal{N}_i} w_{i,j} \mathbf{r}_j(t), \quad i \in \mathcal{F},
\end{equation}
where $\mathbf{r}_j(t)$ denotes the \textit{real-time actual position} of neighbor $j$ obtained from the Vicon motion capture system, and $w_{i,j}$ are fixed communication weights computed from the initial configuration using barycentric coordinates (Eq.~\eqref{folcomweights}).

The neighbor topology for each follower is defined as:
\begin{itemize}
    \item cf3: neighbors $\{$cf1, cf4, cf5$\}$
    \item cf4: neighbors $\{$cf2, cf3, cf5$\}$
    \item cf5: neighbors $\{$cf3, cf4, cf6$\}$
\end{itemize}

The coordination architecture and the sensing are separate. The control law is decentralized: each follower uses only the measured positions of its three neighbors and fixed weights, and no follower receives leader trajectories, transformation parameters, or global reference signals. The motion capture system supplies each quadcopter's position measurement in place of onboard state estimation, the role that GPS, ultra-wideband ranging, or visual-inertial odometry plays in a field deployment, and it takes no part in the coordination logic.

\subsection{Feedback control implementation}

Both leaders and followers employ a PD feedback control law with feedforward acceleration:
\begin{equation}
\mathbf{u}_i = -\mathbf{K}_p (\mathbf{r}_i - \mathbf{r}_{i,d}) - \mathbf{K}_d \dot{\mathbf{r}}_i + \mathbf{a}_{ff},
\end{equation}
where $\mathbf{K}_p$ and $\mathbf{K}_d$ are diagonal gain matrices. For leaders, $\mathbf{K}_p = \text{diag}(2.5, 2.5, 4.0)$ and $\mathbf{K}_d = \text{diag}(1.5, 1.5, 2.0)$, with feedforward acceleration from the pre-computed trajectory. For followers, $\mathbf{K}_p = \text{diag}(3.5, 3.5, 5.0)$ and $\mathbf{K}_d = \text{diag}(2.0, 2.0, 2.5)$, with zero feedforward since only position constraints are available from the weighted neighbor sum.

The gains were selected with the cascaded structure of the architecture in mind. Follower gains are set higher than leader gains because a follower tracks a reference constructed from the measured positions of its neighbors rather than a smooth precomputed trajectory: the reference carries measurement noise and lags the leader motion, and stiffer feedback shortens this lag. Vertical gains are higher than horizontal gains on all vehicles to reject the altitude coupling induced by aerodynamic effects during close-proximity flight. Starting from the default Crazyswarm2 position-loop gains, both sets were adjusted in hover and low-speed tests until the responses were well damped without oscillation.

The control loop operates at 50~Hz, with each drone receiving full-state commands (position, velocity, acceleration, yaw, angular velocity) via the Crazyswarm2 \texttt{cmdFullState} interface over a 2.4~GHz Crazyradio PA link.

\subsection{Experimental protocol}

The experiment was organized into five sequential phases with a total duration of 30 seconds, as illustrated in Fig.~\ref{fig:formation_phases}:

\begin{enumerate}
    \item \textbf{Pre-AT Phase (0--5s):} All drones execute a controlled takeoff to the nominal altitude of 0.75~m.

    \item \textbf{AT Phase 1 -- Pure Contraction (5--10s):} The formation undergoes uniform scaling with $\lambda_1 = \lambda_2$ decreasing from 1.0 to 0.5.

    \item \textbf{AT Phase 2 -- Rigid Body Motion (10--25s):} Translation and rotation are applied while maintaining the contracted shape.

    \item \textbf{AT Phase 3 -- Precise Deformation (25--30s):} Non-uniform scaling with $\lambda_1 \to 0.6$ and $\lambda_2 \to 0.9$ achieves the final configuration.

    \item \textbf{Post-AT Phase:} Synchronized landing sequence.
\end{enumerate}

Transitions between AT phases use a quintic polynomial $\beta(s) = 6s^5 - 15s^4 + 10s^3$ to ensure $C^2$ continuity.

\begin{figure}[t]
    \centering
    \begin{tabular}{ccc}
        \includegraphics[width=0.32\textwidth]{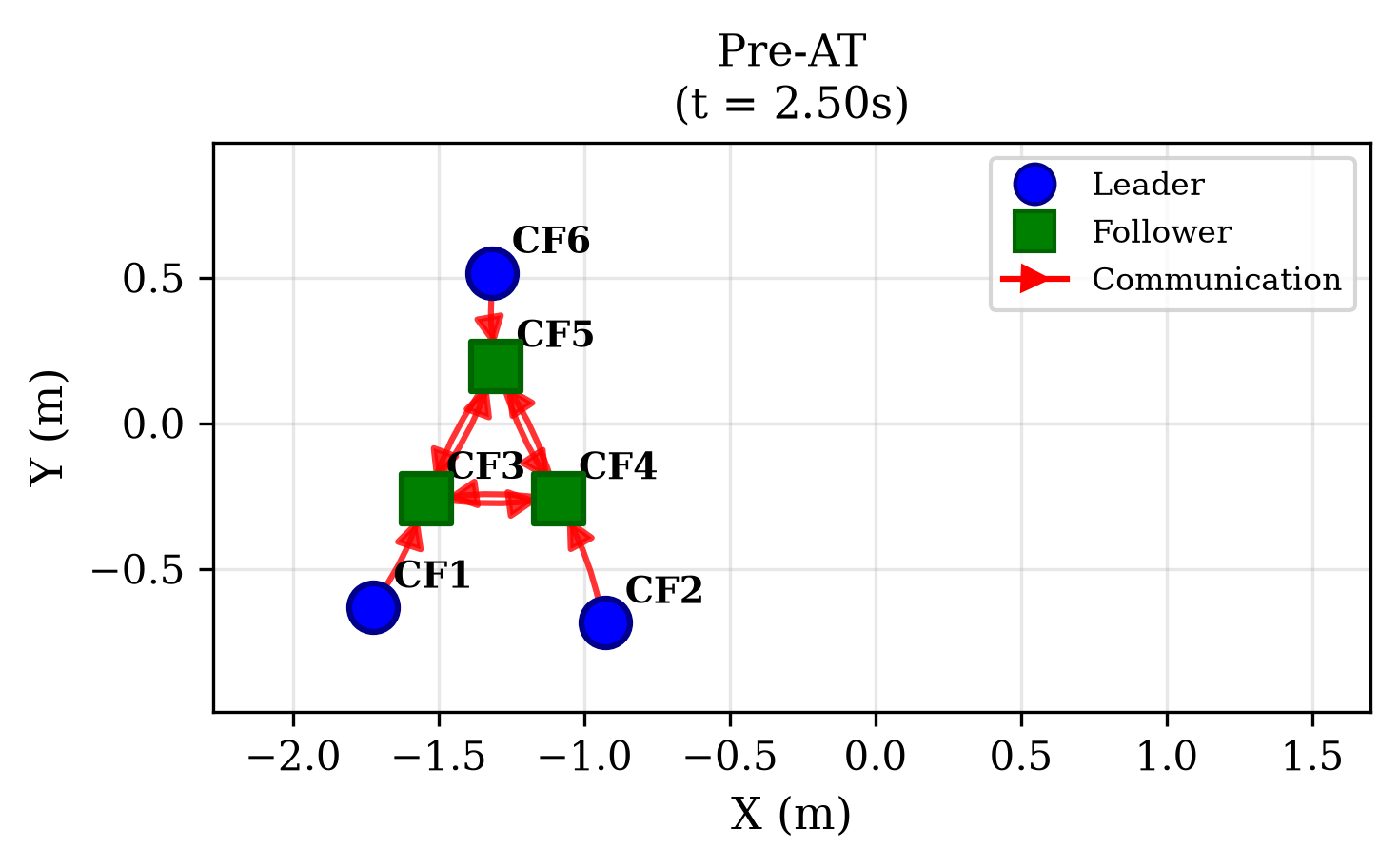} &
        \includegraphics[width=0.32\textwidth]{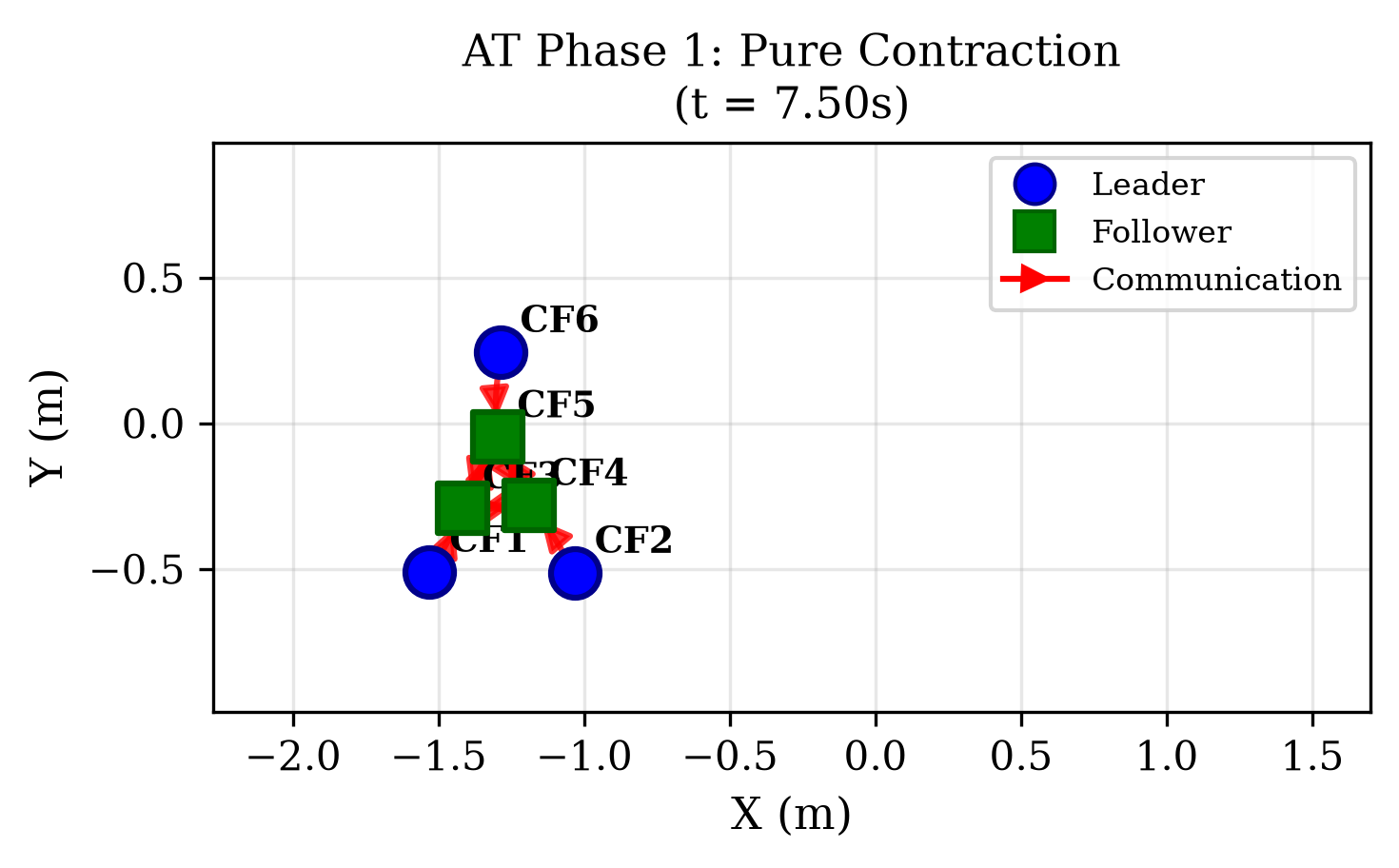} &
        \includegraphics[width=0.32\textwidth]{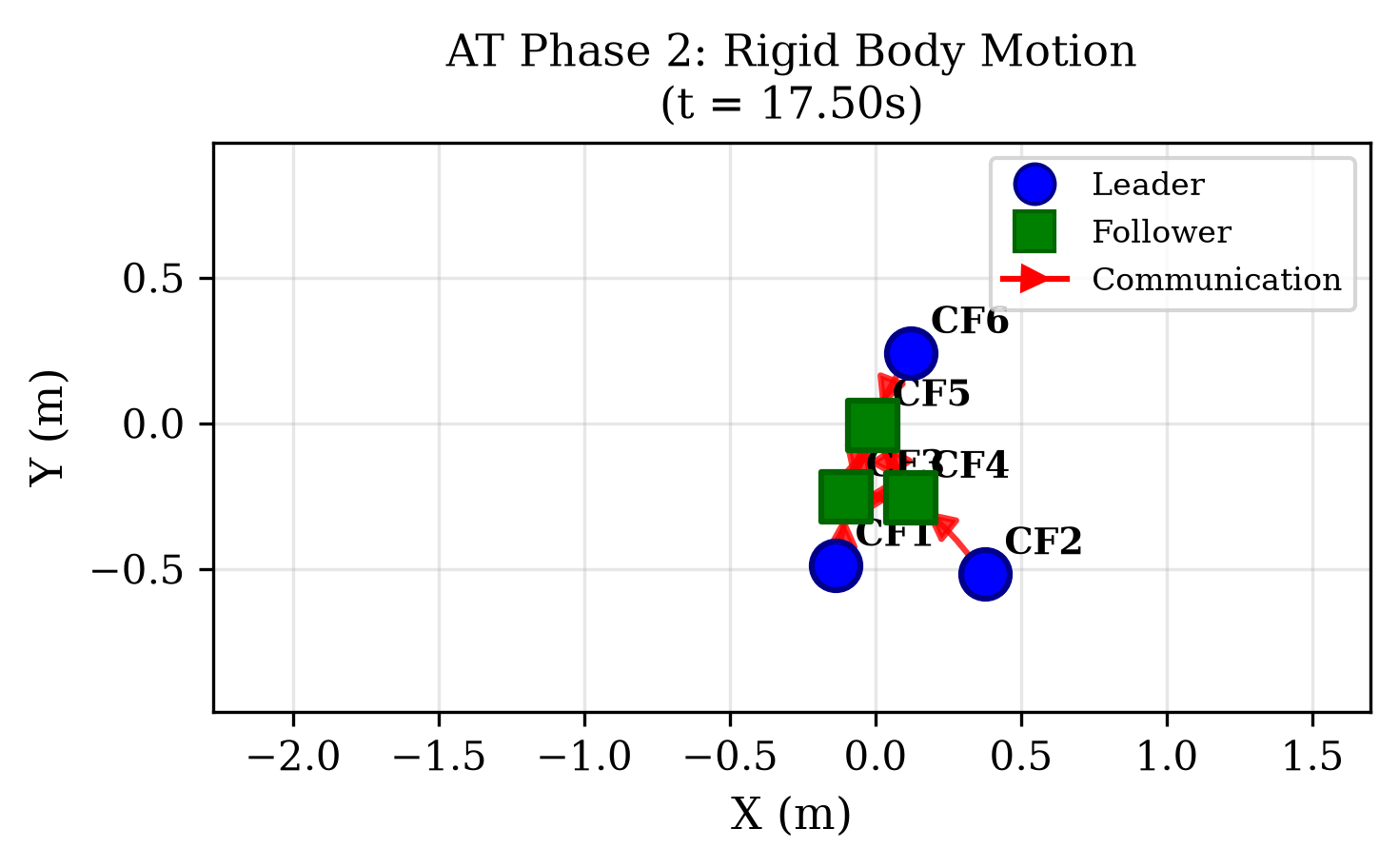} \\
        (a) Pre-AT ($t = 2.5$s) & (b) AT Phase 1 ($t = 7.5$s) & (c) AT Phase 2 ($t = 17.5$s) \\[0.3cm]
        \multicolumn{3}{c}{
            \includegraphics[width=0.32\textwidth]{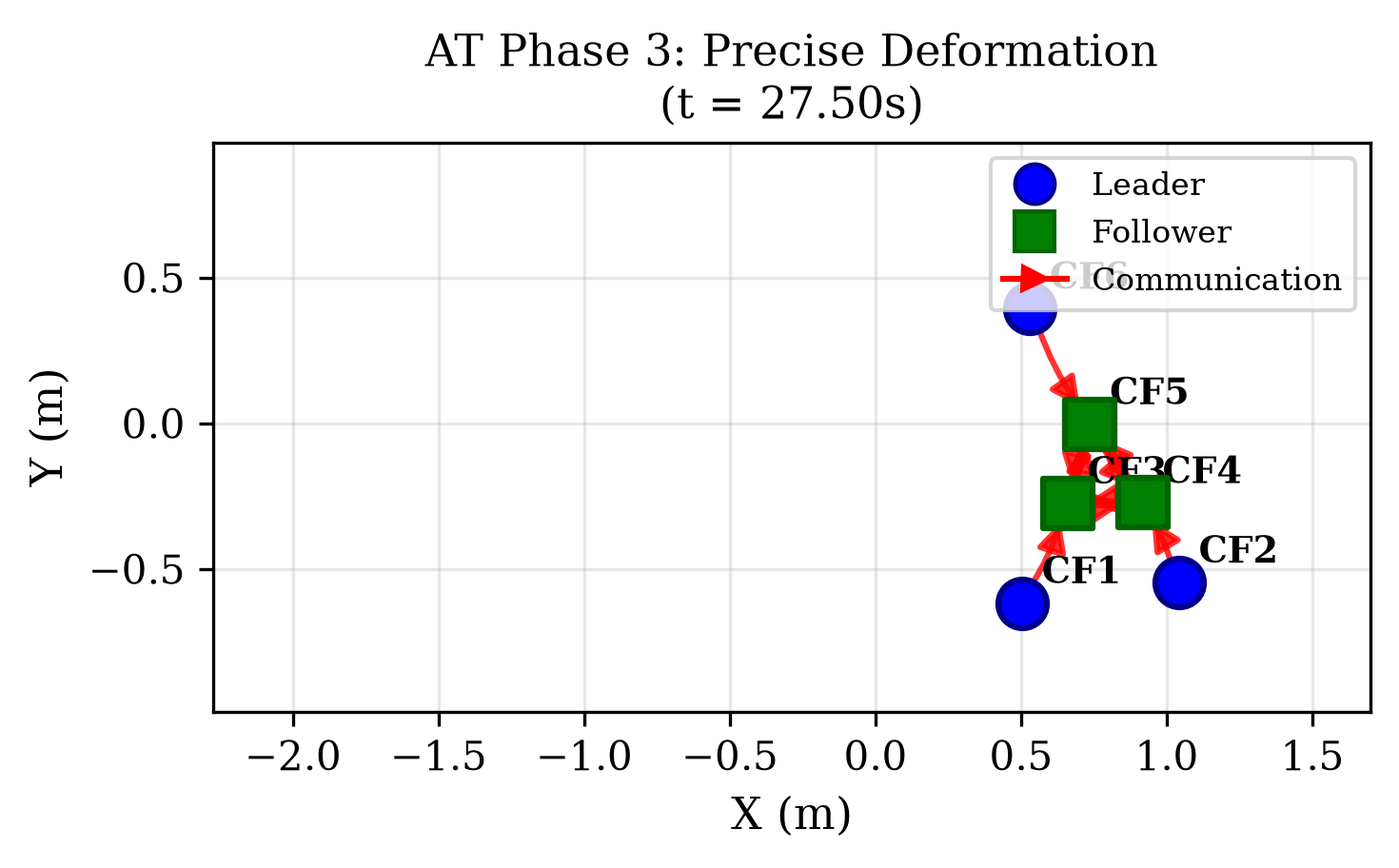} \quad
            \includegraphics[width=0.32\textwidth]{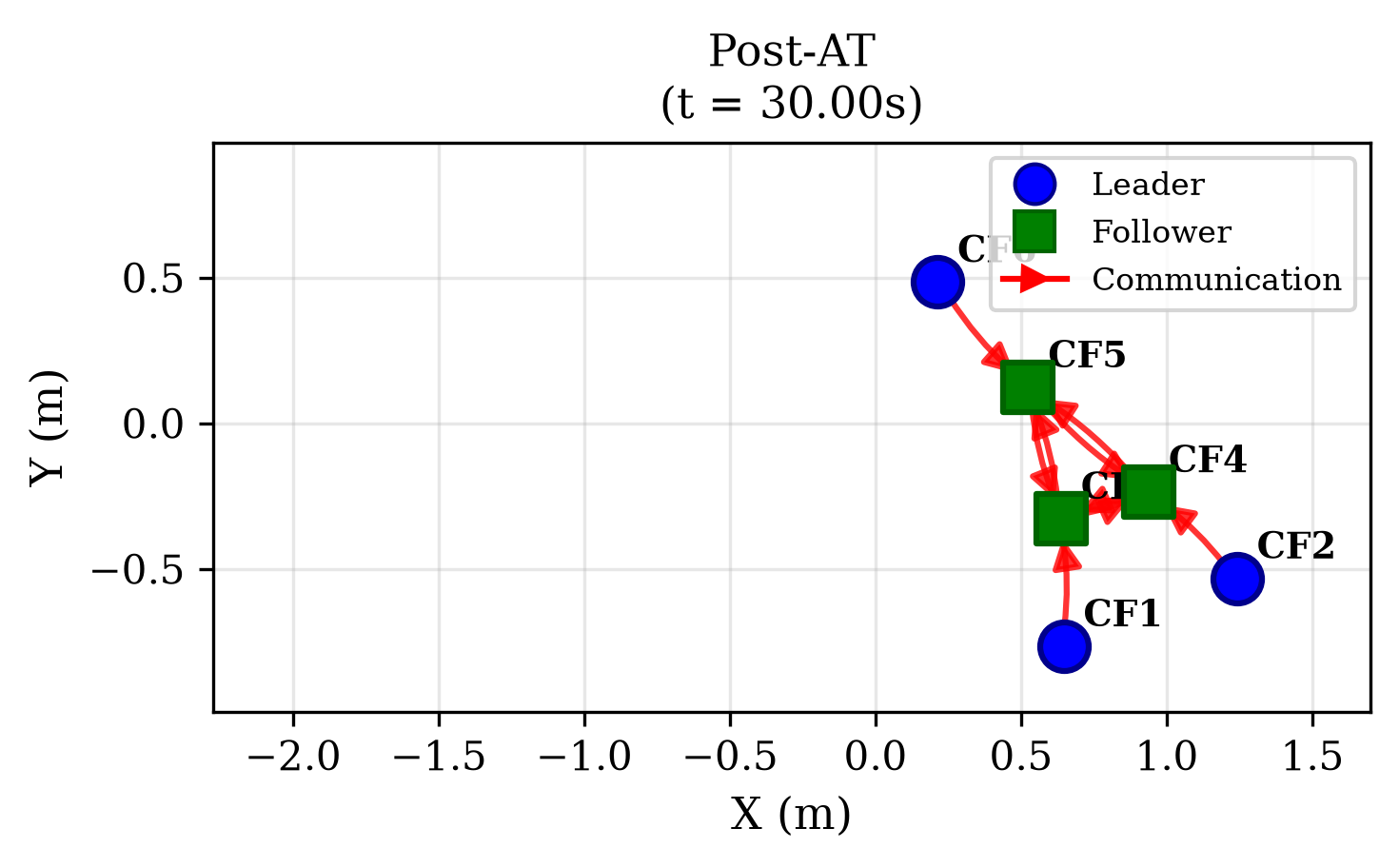}
        } \\
        \multicolumn{3}{c}{(d) AT Phase 3 ($t = 27.5$s) \hspace{1.5cm} (e) Post-AT ($t = 30.0$s)}
    \end{tabular}
    \caption{Formation configurations at each experimental phase. Leaders (blue circles) and followers (green squares) are shown with their communication topology (red arrows). The formation contracts during AT Phase 1, undergoes rigid body motion in AT Phase 2, and deforms non-uniformly in AT Phase 3.}
    \label{fig:formation_phases}
\end{figure}

\begin{table}[t]
\centering
\small
\setlength{\tabcolsep}{4pt}
\renewcommand{\arraystretch}{1.1}
\begin{tabular}{|l|c|c|c|c|c|c|}
\hline
\textbf{AT Phase} & $t_0$ & $t_f$ & $\lambda_{1,0}$ & $\lambda_{1,f}$ & $\lambda_{2,0}$ & $\lambda_{2,f}$ \\
\hline
Pure Contraction & 5s & 10s & 1.00 & 0.50 & 1.00 & 0.50 \\
Rigid Body Motion & 10s & 25s & 0.50 & 0.50 & 0.50 & 0.50 \\
Precise Deformation & 25s & 30s & 0.50 & 0.60 & 0.50 & 0.90 \\
\hline
\end{tabular}
\caption{AT phase parameters showing principal strain evolution.}
\label{tab:at_phases}
\end{table}

\subsection{Environment}

A narrow passage was constructed from stacked storage containers to demonstrate the formation's ability to contract and navigate through a physically constrained corridor (Fig.~\ref{fig:experiment_photos}). The passage width was chosen such that the initial formation could not pass through without contraction, requiring the principal strains to decrease to $\lambda_1 = \lambda_2 = 0.5$ during AT Phase~1. During AT Phase~2, the contracted formation translates through the passage via rigid body motion. A video of the experiment is available online.\footnote{\url{https://www.youtube.com/watch?v=jtAYp5TxVp0}}

\begin{figure}[t]
    \centering
    \begin{tabular}{cc}
        \includegraphics[width=0.48\textwidth]{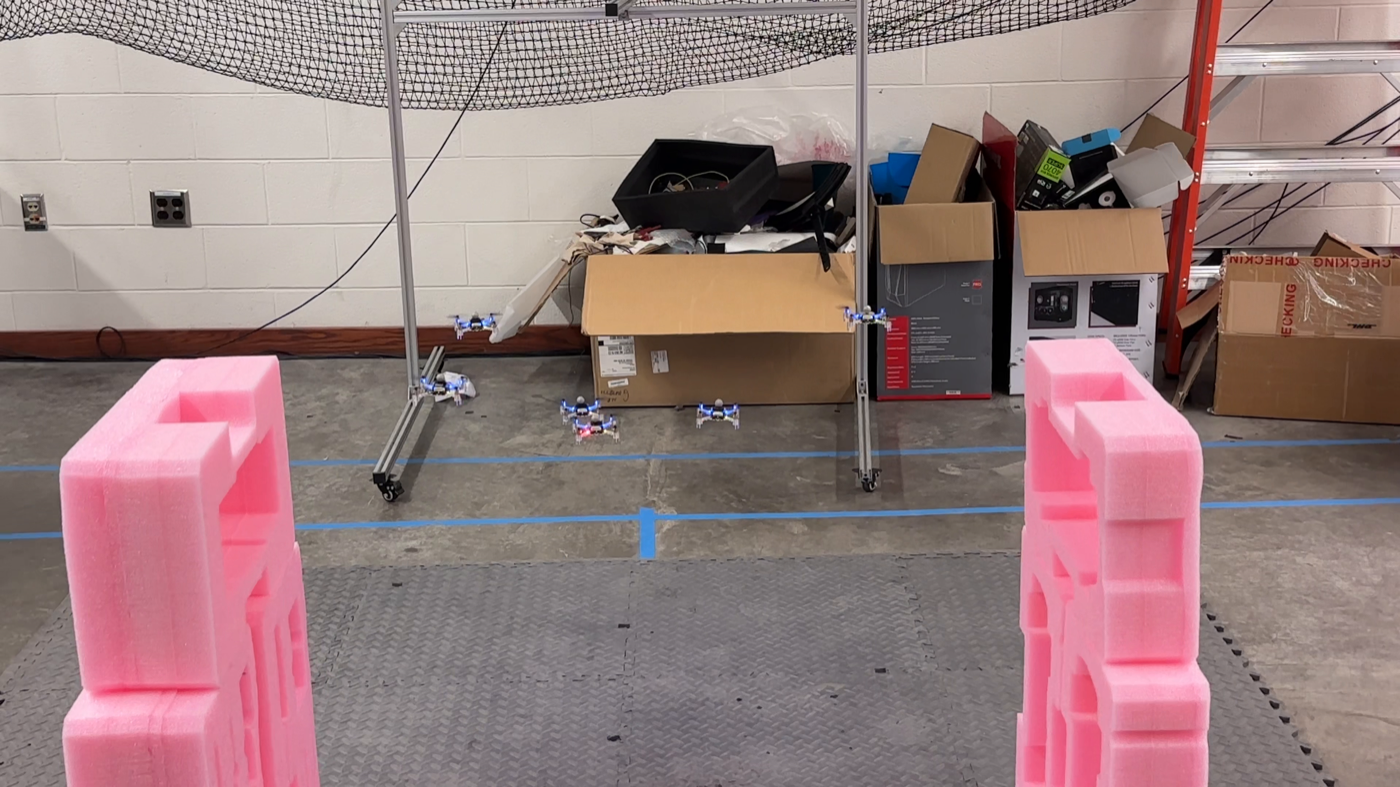} &
        \includegraphics[width=0.48\textwidth]{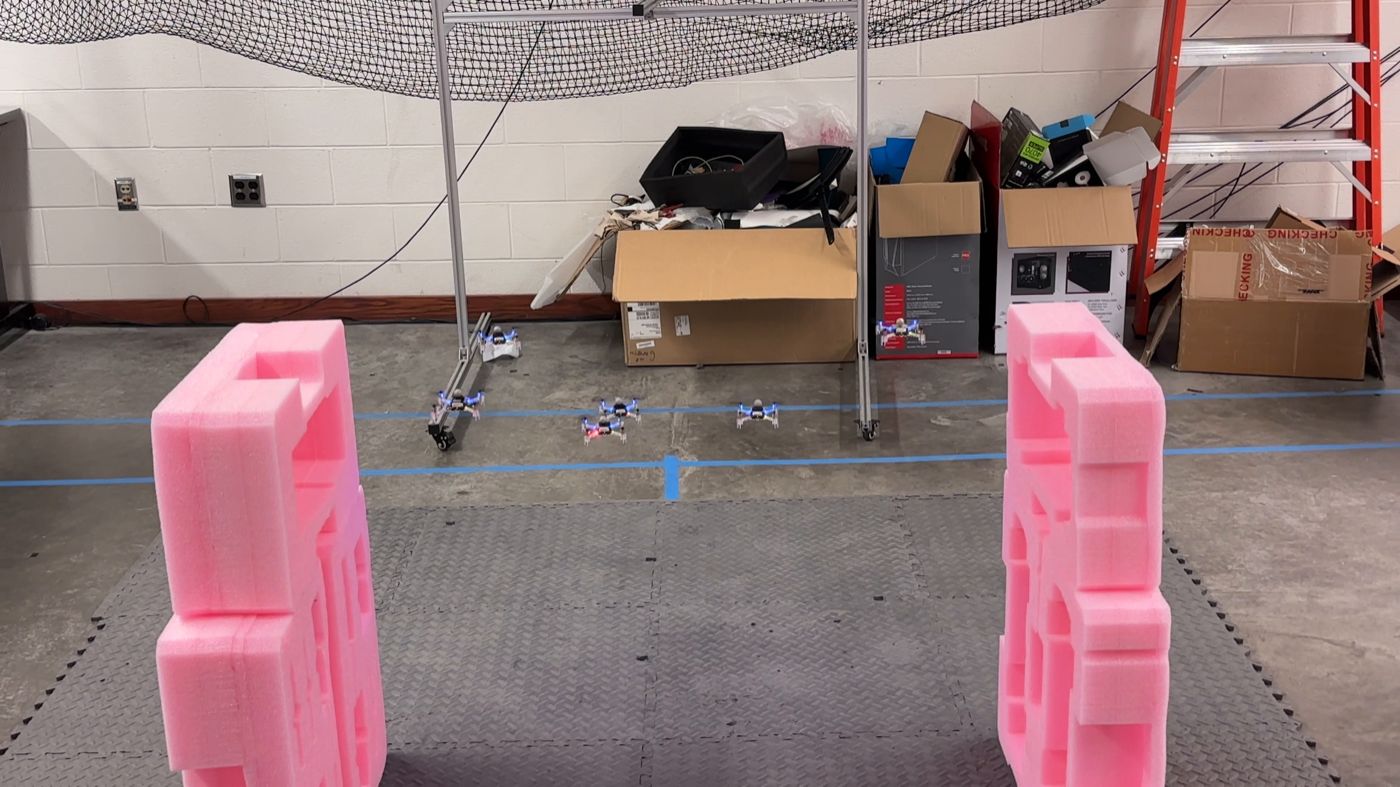} \\
        (a) Contracted formation approaching passage & (b) Formation passing through passage \\[0.3cm]
        \includegraphics[width=0.48\textwidth]{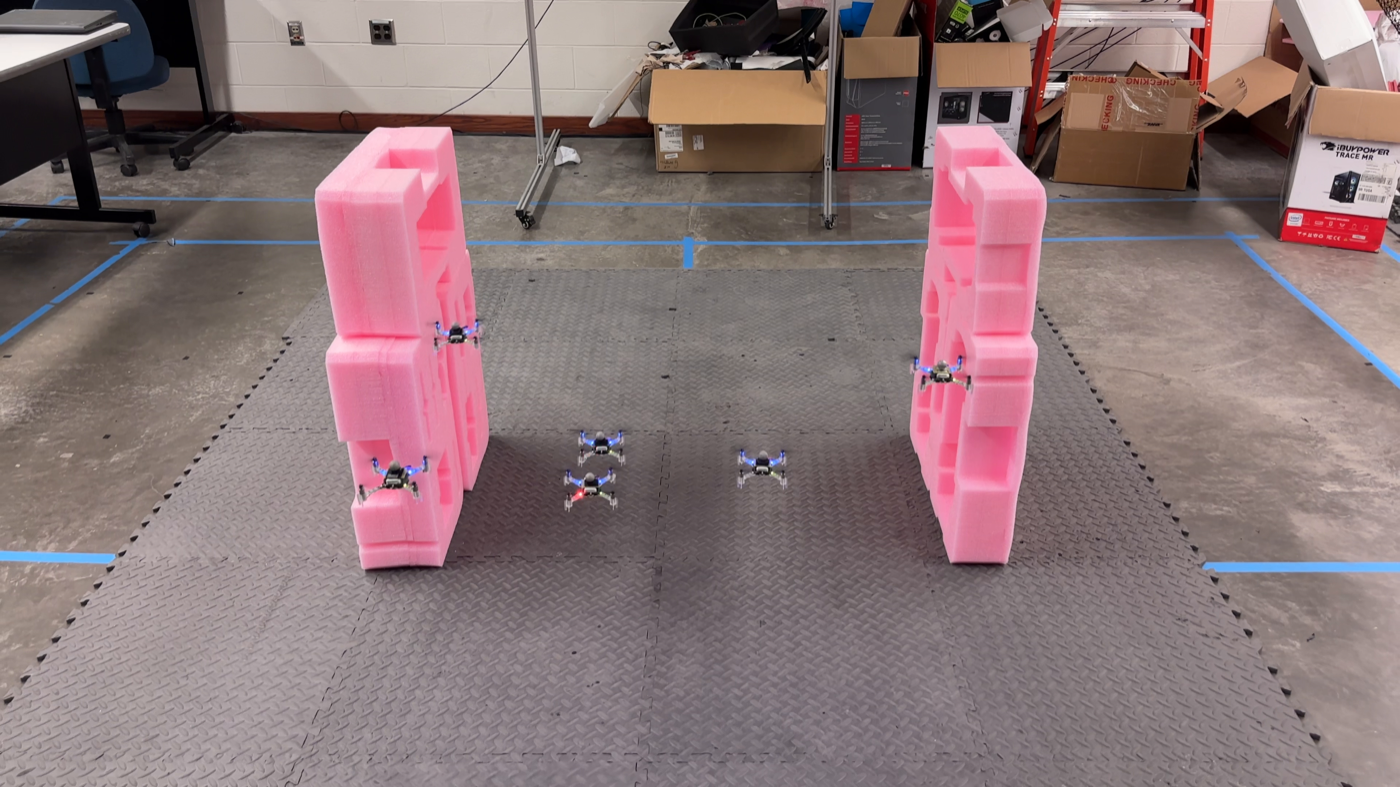} &
        \includegraphics[width=0.48\textwidth]{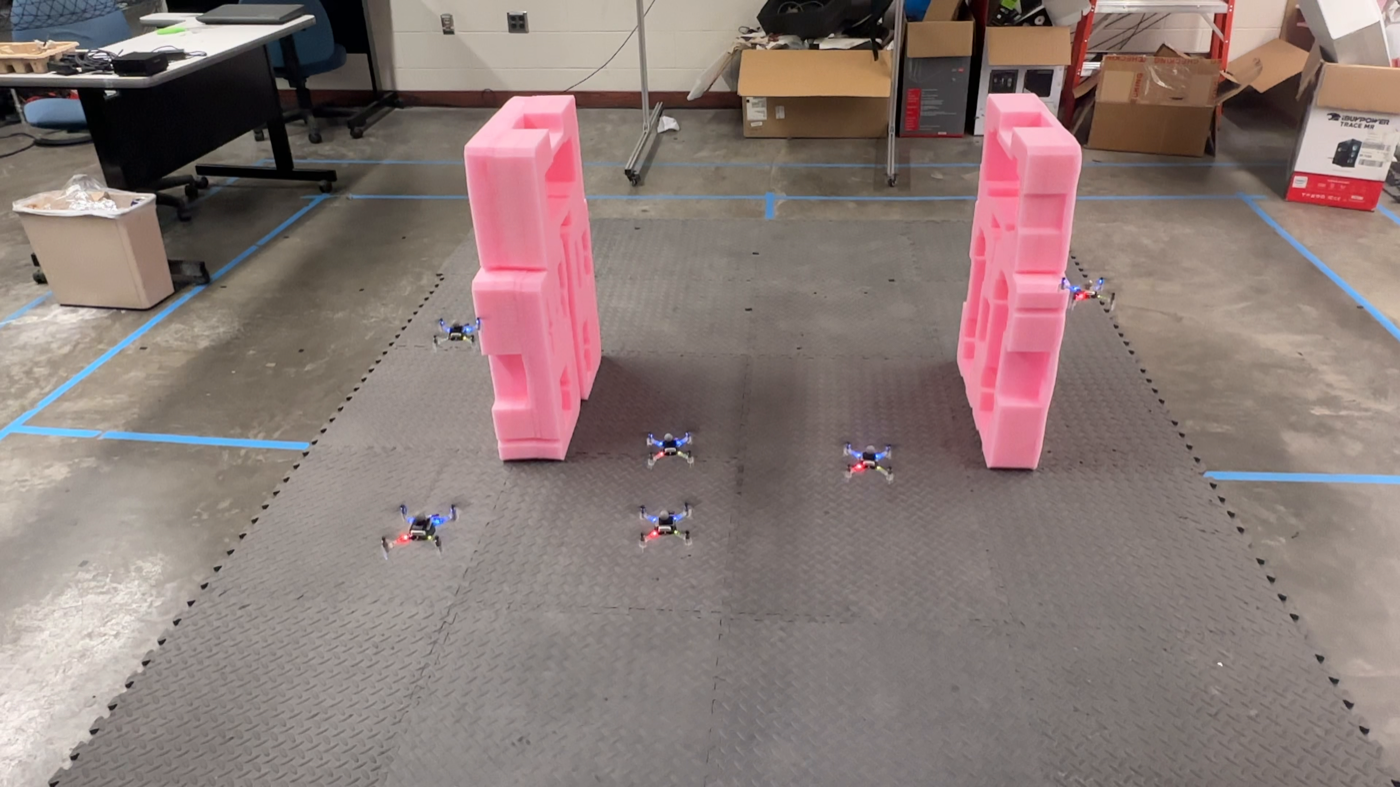} \\
        (c) Formation emerging from passage & (d) Final deformed formation
    \end{tabular}
    \caption{Experimental flight sequence showing the six-drone formation navigating through a narrow passage formed by stacked storage containers. The formation contracts during AT Phase~1 to fit through the passage, translates through during AT Phase~2, and deforms non-uniformly during AT Phase~3.}
    \label{fig:experiment_photos}
\end{figure}

\section{Experimental results}\label{sec:results}

\subsection{Tracking performance}

Fig.~\ref{fig:tracking_errors} shows the tracking errors of all six drones, and Table~\ref{tab:tracking_stats} gives per-drone statistics for this flight. Leader RMSE ranges from 3.4~cm (cf6) to 5.9~cm (cf2) and follower RMSE from 6.7~cm (cf3) to 8.3~cm (cf4). Across the three valid flights the mean leader and follower RMSE were $5.7\pm1.6$~cm and $8.4\pm1.5$~cm; the two flights with matched configuration agree to within 0.1--0.2~cm, so the spread reflects configuration differences between flights rather than run-to-run variability. Followers track less tightly than leaders because each follower's reference is the weighted sum of its neighbors' measured positions, which carries sensor noise and lags the leader motion. The leaders follow smooth precomputed trajectories.

\begin{figure}[t]
    \centering
    \includegraphics[width=\textwidth]{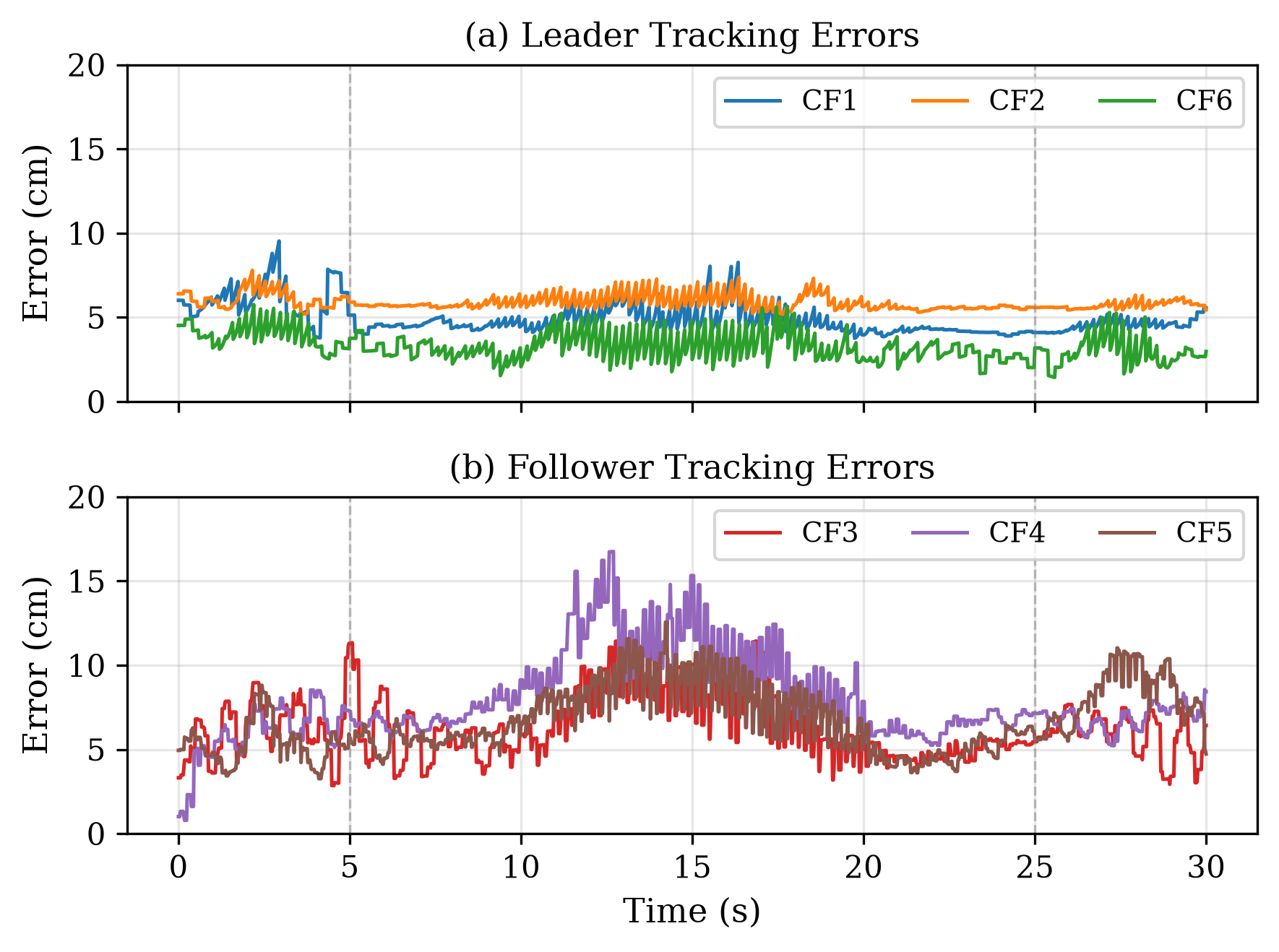}
    \caption{Tracking errors over time. (a) Leaders track pre-computed AT trajectories. (b) Followers track weighted sums of actual neighbor positions per Eq.~\eqref{desiredtraject}.}
    \label{fig:tracking_errors}
\end{figure}

\begin{table}[t]
\centering
\begin{tabular}{|c|c|c|c|c|}
\hline
\textbf{Drone} & \textbf{Role} & \textbf{Mean (cm)} & \textbf{Max (cm)} & \textbf{RMSE (cm)} \\
\hline
cf1 & Leader & 4.90 & 9.53 & 4.98 \\
cf2 & Leader & 5.89 & 7.79 & 5.90 \\
cf6 & Leader & 3.28 & 5.76 & 3.38 \\
cf3 & Follower & 6.39 & 11.58 & 6.69 \\
cf4 & Follower & 7.86 & 16.75 & 8.26 \\
cf5 & Follower & 6.78 & 12.56 & 7.07 \\
\hline
\end{tabular}
\caption{Tracking error statistics from experimental data.}
\label{tab:tracking_stats}
\end{table}

\subsection{Formation trajectories}

Fig.~\ref{fig:formation_traj} illustrates the 2D trajectories in the XY plane, comparing desired and actual paths. The formation successfully executes the contraction, rigid body motion, and deformation phases while maintaining the topological structure defined by the communication graph.

\begin{figure}[t]
    \centering
    \includegraphics[width=\textwidth]{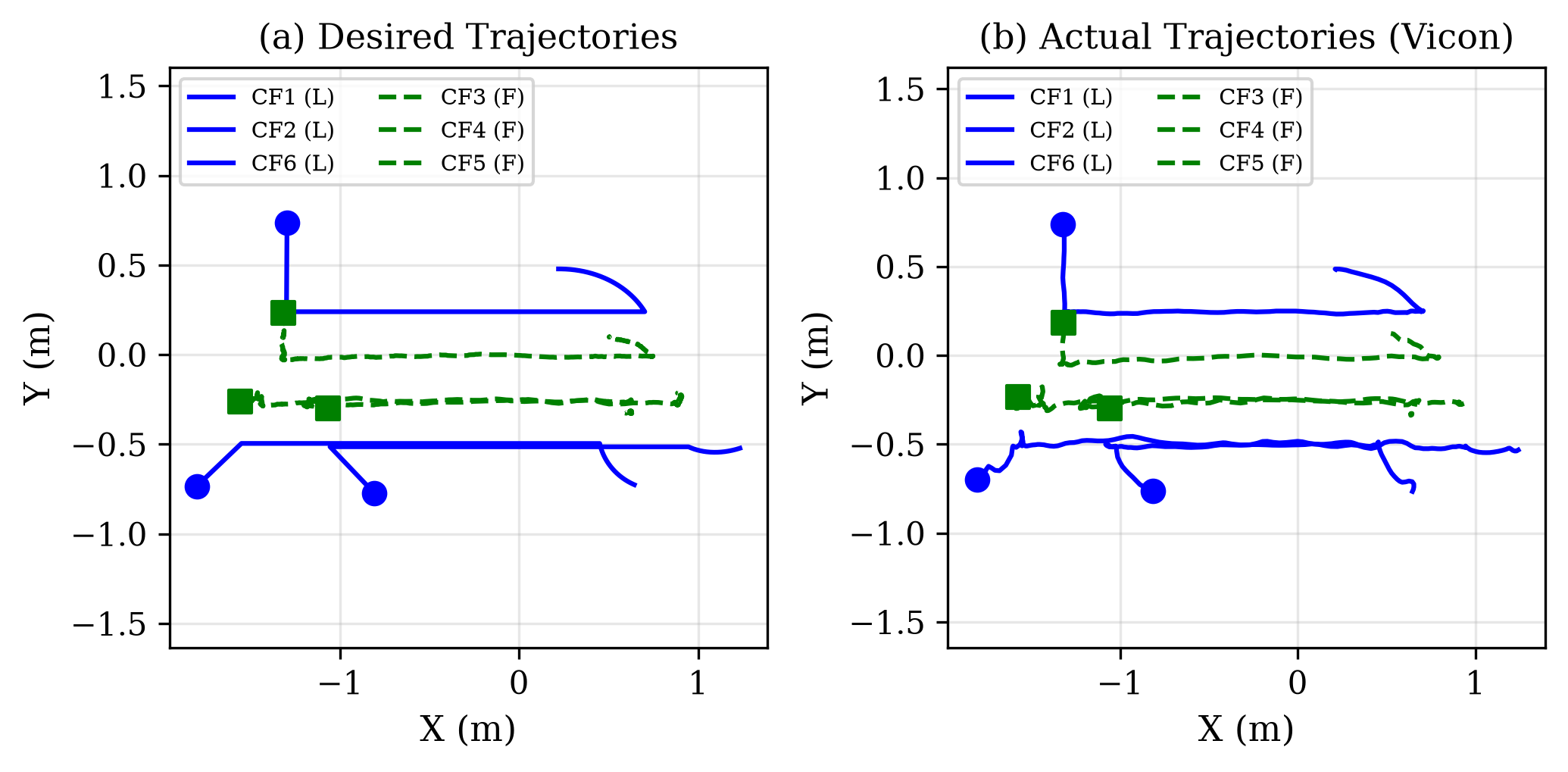}
    \caption{Formation trajectories in the XY plane. (a) Desired trajectories from the AT framework. (b) Actual trajectories measured by Vicon. Circles denote initial positions.}
    \label{fig:formation_traj}
\end{figure}

\subsection{Validation of decentralized AT}

The experimental results validate the key theoretical claims:

\begin{enumerate}
    \item \textbf{Decentralized Computation:} Followers successfully track their desired positions using only local neighbor information, without access to leader trajectories or global coordination.

    \item \textbf{Real-Time Position Updates:} The control law uses actual measured positions $\mathbf{r}_j(t)$ rather than pre-computed trajectories, demonstrating true decentralized operation.

    \item \textbf{Collision Avoidance:} We computed the distance between every pair of quadcopters at every logged time step. The smallest separation over the whole flight was $0.154$~m, between cf3 and cf4 at $t = 12.7$~s, above the $2r = 0.13$~m collision threshold, and no pair crossed the threshold at any time (Fig.~\ref{fig:separation}). Theorem~\ref{thm:safety} gives a sufficient condition $\lambda \geq \lambda_{\min} = 2(\delta + r)/d_{\min}$ for this separation, with $d_{\min} = 0.50$~m the minimum initial inter-agent distance (cf3 and cf4, Table~\ref{tab:initial_positions_updated}) and $\delta$ a bound on the tracking error. The bound is conservative: it assumes the two nearest agents deviate toward each other by the full $\delta$ at maximal contraction, which places $\lambda_{\min}$ above the flown $\lambda = 0.5$ for any $\delta$ consistent with the measured errors (Section~\ref{sec:discussion}). A collision depends on the relative deviation between neighbors, and the mean residual of a least-squares affine fit of the realized formation onto its initial configuration was 4.4~cm, so the team held the commanded relative geometry closely and the measured separation stayed inside the threshold.
\end{enumerate}

The source code and experimental data are available at:
\begin{center}
\url{https://github.com/smart-lab-uofa/affine-formation-control}
\end{center}

\begin{figure}[t]
    \centering
    \includegraphics[width=0.8\textwidth]{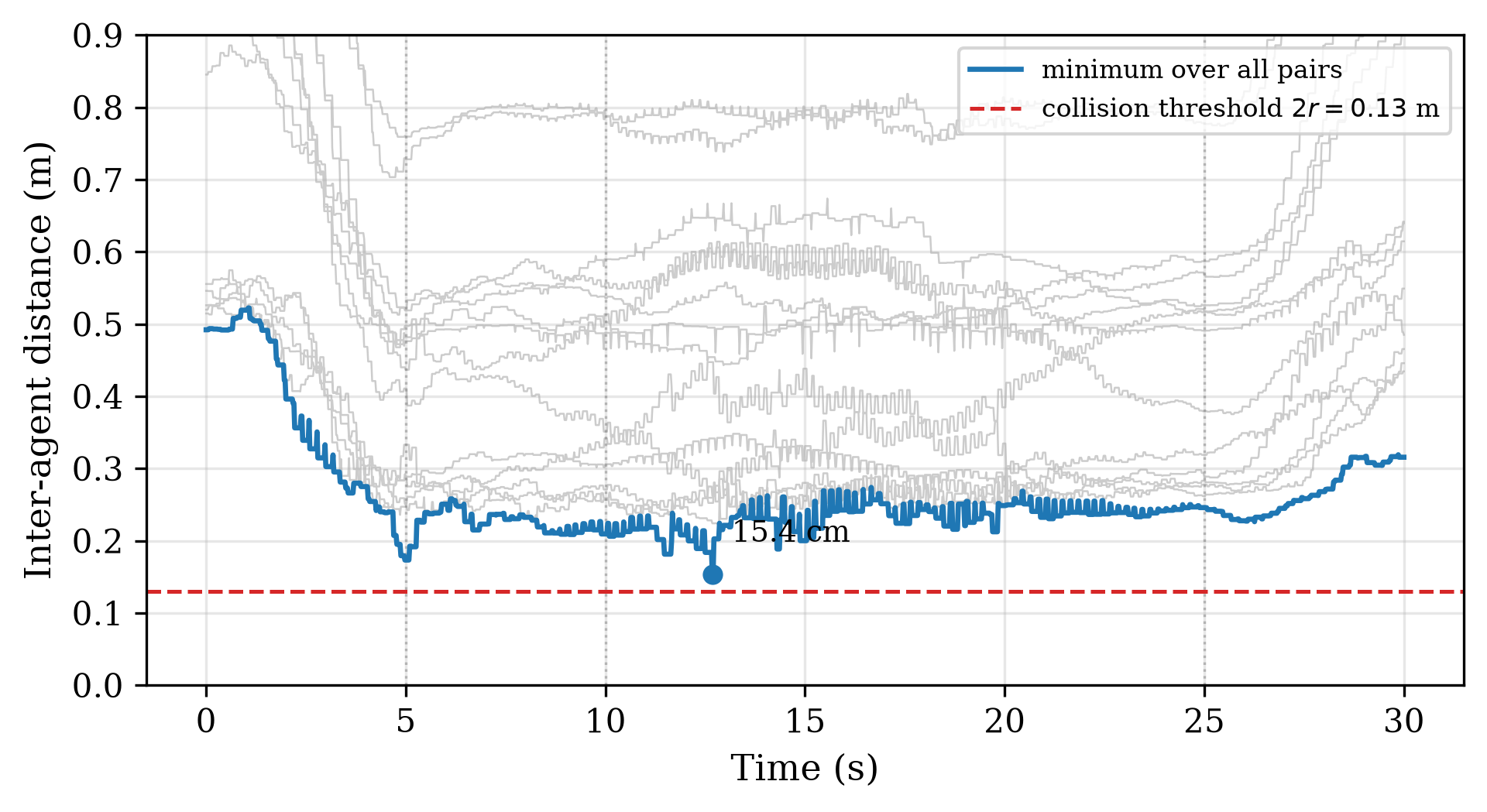}
    \caption{Pairwise inter-agent distance over the flight. The bold curve is the minimum over all 15 pairs and the dashed line is the $2r = 0.13$~m collision threshold. The smallest separation is 0.154~m, at $t = 12.7$~s.}
    \label{fig:separation}
\end{figure}

\subsection{Phase-by-phase tracking analysis}

The tracking error profiles in Fig.~\ref{fig:tracking_errors} show phase-dependent behavior. During the pre-AT takeoff (0--5~s), leader errors settle to a few centimeters as the drones reach their hover positions. At the onset of pure contraction ($t = 5$~s), follower errors rise as their reference positions start to move through the weighted sum of neighbor positions. During the rigid-body motion phase (10--25~s) follower errors are largest, with the peak of 16.8~cm at a direction change in the leader trajectories. The final deformation phase (25--30~s) adds another transient as the non-uniform scaling ($\lambda_1 \to 0.6$, $\lambda_2 \to 0.9$) differentially shifts follower references. Errors stay bounded in every phase (Fig.~\ref{fig:convergence}), consistent with the bounded-error result of Proposition~\ref{prop:bounded}.

Table~\ref{tab:phase_errors} quantifies this. Leader errors change little across phases (mean 4.4--5.4~cm), while follower errors grow from 5.8~cm at takeoff to 7.7~cm during rigid-body motion, the phase with the fastest leaders, and recover partially in the final deformation. The follower maximum, 16.8~cm in AT Phase~2, falls at a direction change of the leader trajectories, matching the lag of the cascaded neighbor-measurement path.

\begin{table}[t]
\centering

\begin{tabular}{|l|c|c|c|c|}
\hline
\textbf{Phase} & \multicolumn{2}{c|}{\textbf{Leaders (cm)}} & \multicolumn{2}{c|}{\textbf{Followers (cm)}} \\
 & Mean & Max & Mean & Max \\
\hline
Pre-AT (0--5~s) & 5.4 & 9.5 & 5.8 & 11.3 \\
AT 1: Contraction (5--10~s) & 4.4 & 6.4 & 6.1 & 11.3 \\
AT 2: Rigid-body motion (10--25~s) & 4.6 & 8.3 & 7.7 & 16.8 \\
AT 3: Deformation (25--30~s) & 4.4 & 6.3 & 7.0 & 11.0 \\
\hline
\end{tabular}
\caption{Tracking error statistics per experimental phase, aggregated over the three leaders and the three followers.}
\label{tab:phase_errors}

\end{table}

\begin{figure}[t]
    \centering
    \includegraphics[width=0.8\textwidth]{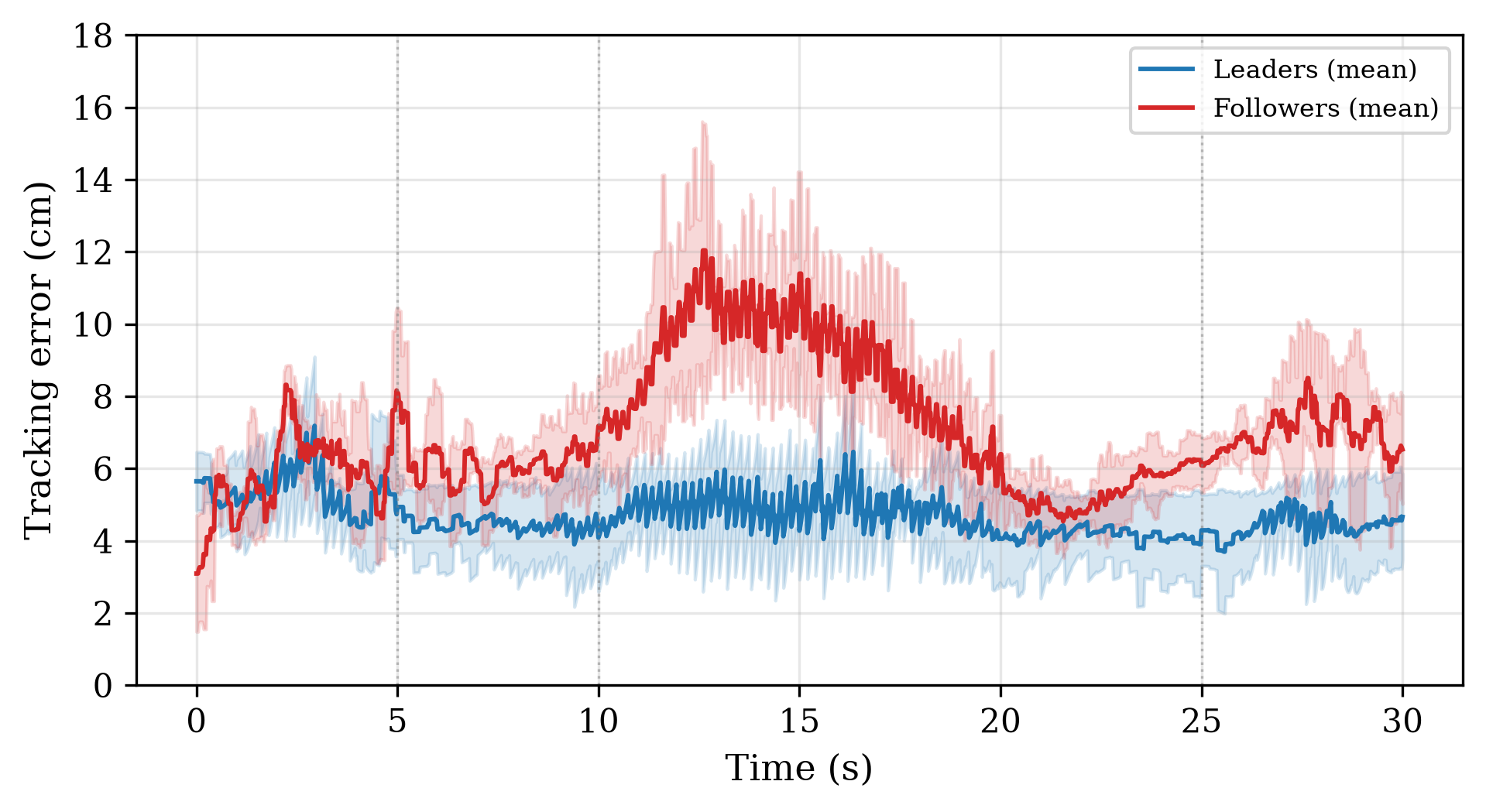}
    \caption{Mean tracking error of the leaders and of the followers over time, with shaded $\pm1$ standard deviation across the three drones in each role. Leader error stays near 5~cm; follower error peaks during the rigid-body phase and remains bounded, consistent with Proposition~\ref{prop:bounded}.}
    \label{fig:convergence}
\end{figure}

\section{Discussion}\label{sec:discussion}

The experimental results demonstrate the practical viability of decentralized AT for multi-agent coordination. The bounded tracking errors observed throughout all phases confirm that the decentralized control law successfully guides followers to their AT-defined positions using only local communication. The error differential between leaders (3--6~cm) and followers (7--8~cm) is attributable to the cascaded nature of the architecture: follower reference signals are computed from the actual (noisy) positions of neighbors, introducing an additional error source compared to the leaders' smooth pre-computed trajectories.

The safety condition is conservative. Theorem~\ref{thm:safety} is a sufficient condition built from worst-case bounds, so it trades tightness for a guarantee that holds for any team size through a single scalar check. The flight data shows that the resulting margin is practical: the smallest measured inter-agent distance was $0.154$~m against the $0.13$~m threshold, and the mean affine-fit residual was 4.4~cm. This conservatism is acceptable in safety-critical operation and in large teams, where maintaining and verifying pairwise constraints becomes infeasible; when a mission requires tighter spacing than the strain bound permits, a per-pair analysis of the specific configuration is the appropriate refinement.

\subsection{Comparison with related approaches}
% per your guidance that consensus/containment are not task-comparable with AT.
% Please review the argument and the table entries.]
A direct experimental comparison between decentralized AT and the classical coordination paradigms is not well defined, because the methods pursue different objectives. Consensus-based formation control drives the agents to a rigid formation specified by fixed displacements and cannot produce the deformations that the task requires: contraction, shear, and non-uniform scaling have no counterpart in the consensus objective. Containment control keeps the followers inside the convex hull of the leaders but does not assign them unique positions, so tracking error with respect to a desired configuration is not defined. Neither method can execute the narrow-passage mission of Section~\ref{sec:experimental_setup}. For this reason, Table~\ref{tab:comparison} compares the methods analytically, on the indices that are well defined for all of them: the convergence condition, the follower objective, the deformation capability, the collision-avoidance guarantee, and the computation and communication cost per follower per control step. Decentralized AT matches the local-communication cost of consensus and containment while assigning unique follower positions, supporting full formation deformation, and providing a collision-avoidance guarantee through a single strain bound rather than pairwise constraints that grow with the number of agents. On the hardware, the decentralized update took 0.86~ms per agent per control step, and the loop held 50~Hz at a 26\% duty cycle (Fig.~\ref{fig:computation}), matching the $O(1)$ per-follower cost in Table~\ref{tab:comparison}.

\begin{figure}[t]
    \centering
    \includegraphics[width=0.7\textwidth]{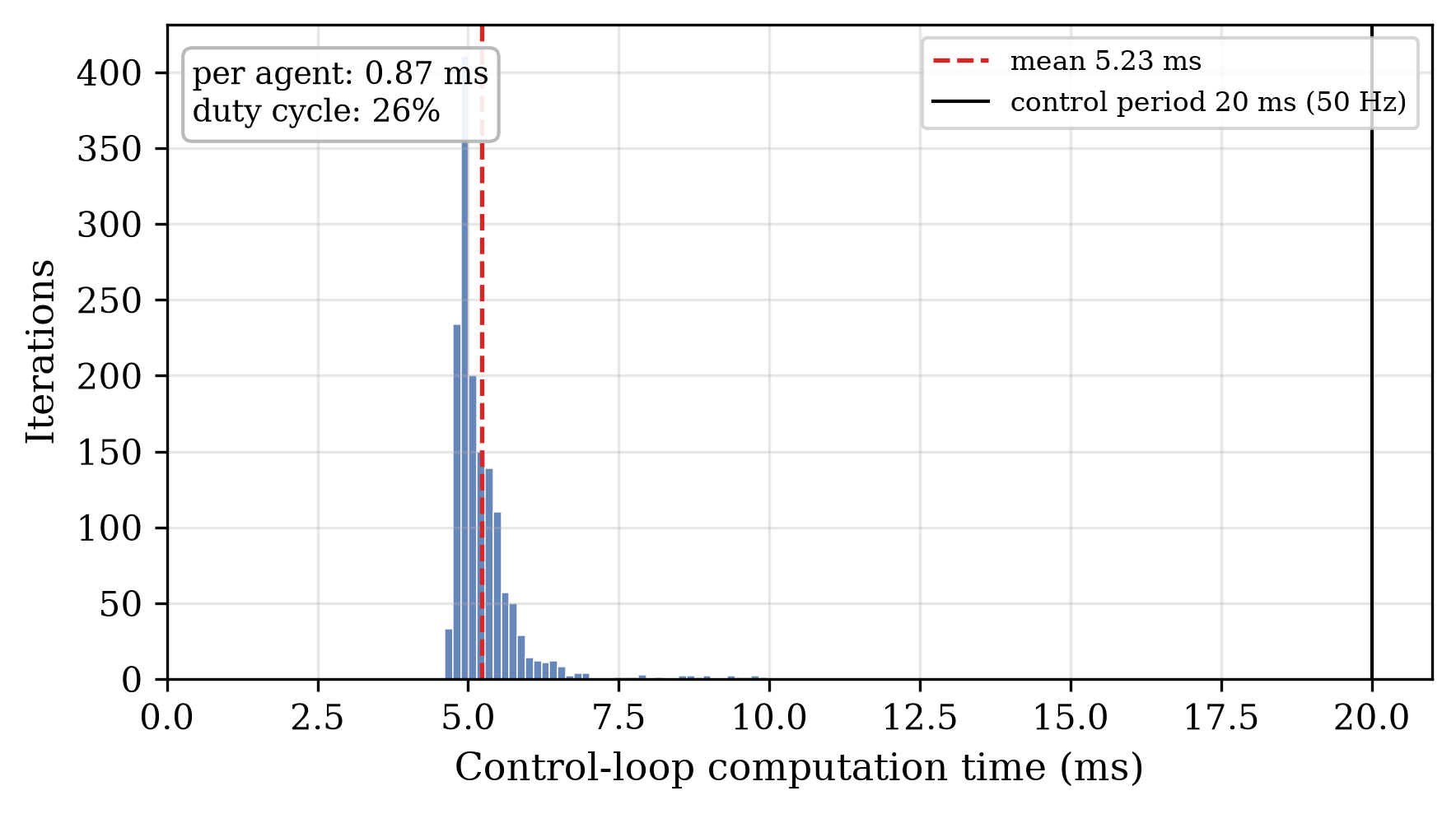}
    \caption{Per-iteration computation time of the control loop over the flight. The mean is 5.23~ms against the 20~ms control period, a 26\% duty cycle, or 0.86~ms per agent.}
    \label{fig:computation}
\end{figure}

\subsection{Limitations and future work}
Several limitations of the present study should be noted. First, while the control algorithm is fully decentralized, with each follower computing its desired position from only its neighbors' measured positions, the current implementation relies on a centralized Vicon motion capture system for state estimation and a ground station for command relay via Crazyradio. Achieving fully onboard decentralized sensing and communication remains an important direction for future work.

Second, the experiments were conducted with six agents (three leaders, three followers). While the per-follower computational cost is $O(1)$, validating the framework with significantly larger teams would provide stronger empirical evidence of scalability.

Finally, the present work is restricted to two-dimensional AT. Extension to three-dimensional AT, requiring four leaders positioned at the vertices of a tetrahedron, would broaden the applicability to more complex coordination scenarios. In the 3D formulation, the barycentric weights are computed from the 3D initial configuration, and a third principal strain $\lambda_3$ enters the collision-avoidance bound. Three technical problems remain open before such a deployment: the downwash coupling between vertically neighboring quadcopters, which perturbs followers located below other vehicles; the computation of communication weights for followers located near the faces of the leading tetrahedron, where the weight matrix approaches singularity; and the planning of the additional deformation coordinates introduced by the third dimension. We plan to extend the decentralized realization to three-dimensional affine transformations, following the $n$-D simplex formulation of \cite{rastgoftar2021safe}, with collision avoidance enforced by constraining the eigenvalues of the affine map.

\begin{table}[t]
\centering
\small
\setlength{\tabcolsep}{3pt}
\renewcommand{\arraystretch}{1.15}

\resizebox{\textwidth}{!}{%
\begin{tabular}{|l|c|c|c|}
\hline
 & \textbf{Consensus} & \textbf{Containment} & \textbf{Decentralized AT} \\
\hline
Follower objective & Rigid formation & Leaders' convex hull & Unique AT position \\
Convergence condition & \begin{tabular}{@{}c@{}}Spanning tree\\from leader\end{tabular} & \begin{tabular}{@{}c@{}}Path from a leader\\to each follower\end{tabular} & \begin{tabular}{@{}c@{}}Path from every leader\\to every follower\end{tabular} \\
Formation deformation & None (rigid) & Not specified & Full \\
Collision avoidance & No guarantee & No guarantee & Global, one strain bound \\
Computation per step & $O(|\mathcal{N}_i|)$ & $O(|\mathcal{N}_i|)$ & $O(1)$ (3 weights) \\
Communication per step & $|\mathcal{N}_i|$ states & $|\mathcal{N}_i|$ states & 3 neighbor states \\
\hline
\end{tabular}}
\caption{Analytical comparison of decentralized AT with consensus-based formation control and containment control. Costs are per follower per control step; $\mathcal{N}_i$ denotes the neighbor set and $N$ the team size.}
\label{tab:comparison}

\end{table}

\section{Conclusion}\label{Conclusion}

This paper presented a decentralized AT framework for safe multi-agent aerial coordination. The proposed approach represents collective motion through a global affine map and uses the principal strains of this map to impose a compact collision-avoidance condition. This provides a scalable safety certificate that is independent of the number of agents and avoids the pairwise constraint growth associated with many safety-critical coordination methods.
A decentralized leader--follower realization was developed in which followers reconstruct their reference positions using only real-time measured positions of neighboring agents and fixed barycentric communication weights. Under local tracking convergence and fixed terminal leader positions, the resulting closed-loop system converges to a configuration consistent with the prescribed affine transformation.
The framework was validated experimentally using six Crazyflie quadcopters navigating through a physically constrained narrow passage. The experiments demonstrated contraction, rigid-body motion, and non-uniform deformation while maintaining bounded tracking errors and keeping the measured inter-agent separation above the collision threshold. These results establish decentralized affine transformation as a promising mechanism for scalable, safe, and geometry-aware coordination of multi-agent aerial systems. Future work will extend the framework to larger teams, onboard sensing and communication, and fully three-dimensional affine transformations.

%% ============================================================
%% Declarations
%% ============================================================

\section*{CRediT authorship contribution statement}
\textbf{Garegin Mazmanyan:} Experimental implementation, Data analysis, Manuscript preparation. \textbf{Hossein Rastgoftar:} Conceptualization, Theoretical development, Supervision, Manuscript preparation.

\section*{Declaration of competing interest}
The authors declare that they have no known competing financial interests or personal relationships that could have appeared to influence the work reported in this paper.

\section*{Data availability}
The experimental data and source code supporting this study are available at \url{https://github.com/smart-lab-uofa/affine-formation-control}.

\section*{Acknowledgements}
Not applicable.

%% ============================================================
%% References
%% ============================================================

\bibliographystyle{elsarticle-num}
\bibliography{reference}

\end{document}